\documentclass[a4paper,12pt]{extarticle}
\usepackage[utf8]{inputenc}
\usepackage[top=2cm,left=2.5cm,right=2cm,bottom=3cm]{geometry}

\usepackage{amsmath}
\usepackage{amssymb}
\usepackage{amsthm}
\usepackage{mathabx}
\usepackage{mathrsfs}
\usepackage{stmaryrd}
\usepackage[hypertexnames=false]{hyperref}
\usepackage[dvipsnames]{xcolor}
\usepackage{tensor}
\usepackage{subfiles} 

\numberwithin{equation}{section}

\newtheorem{lem}{Lemma}[section]
\newtheorem{thm}[lem]{Theorem}

\numberwithin{equation}{section}

\newcommand{\dd}{\mathrm{d}}
\newcommand{\ddv}{\mathrm{d}_{\mathsf{v}}}
\newcommand{\ddh}{\mathrm{d}_{\mathsf{h}}}
\newcommand{\gh}{\mathrm{gh}}
\newcommand{\PP}{\mathscr{P}}
\newcommand{\DD}{\mathcal{D}}
\newcommand{\Lx}{\mathbb{L}}

\newcommand{\hM}{\mathcal{A}}
\newcommand{\E}{\mathcal{E}}

\newcommand{\F}{\mathcal{F}}
\newcommand{\I}{\mathcal{I}}

\newcommand{\hS}{\mathcal{S}}

\newcommand{\hT}{\mathcal{T}}
\newcommand{\vol}{\mathcal{V}}
\newcommand{\LD}{\mathcal{L}}
\newcommand{\hO}{\Omega}

\newcommand{\sfT}{\mathsf{T}}

\newcommand{\op}{\mathcal{O}}
\newcommand{\pr}{\mathrm{pr}}

\newcommand{\sfa}{\mathsf{a}}
\newcommand{\sfb}{\mathsf{b}}

\newcommand{\bref}[1]{\textbf{\ref{#1}}}

\begin{document}

\thispagestyle{empty}

\begin{center}
   {\Large 
   Ambient-space variational calculus for gauge fields\\ on constant-curvature spacetimes
    \\[15pt]    }
\end{center}

\vspace{5mm}

\begin{center}
    {Xavier~Bekaert$^a$, Nicolas~Boulanger$^b$, Maxim~Grigoriev$^{c,d}$, Yegor Goncharov$^{a,b}$
    }
\end{center}

\vspace{5mm}

\begin{enumerate}\small
\item[${}^a$] Institut Denis Poisson, Unit\'e Mixte de Recherche 7013,\\
Universit\'e de Tours, Universit\'e d'Orl\'eans, CNRS,\\
Parc de Grandmont, 37200 Tours, France

\item[${}^b$] Physique de l’Univers, Champs et Gravitation,\\
Université de Mons - UMONS,\\
Place du Parc 20, 7000 Mons, Belgium

\item[${}^c$] I.E. Tamm Department of Theoretical Physics,\\
P.N. Lebedev Physical Institute, \\
Leninsky ave. 53, 119991 Moscow, Russia

\end{enumerate}

\vspace{2mm}

\begin{center}
\texttt{\small xavier.bekaert@lmpt.univ-tours.fr,\quad nicolas.boulanger@umons.ac.be,\\
grig@lpi.ru,\quad yegor.goncharov@univ-tours.fr}
\end{center}

\vspace{8mm}

\begin{abstract}
    \noindent We propose a systematic generating procedure to construct free Lagrangians for massive, massless and partially massless, totally-symmetric tensor fields on $AdS_{d+1}$ starting from the BRST Lagrangian description of massless fields in the flat ambient space $\mathbb{R}^{d,2}\,$.   A novelty is that the Lagrangian is described by a $d+1$ form on $\mathbb{R}^{d,2}$ whose pullback to $AdS_{d+1}$ gives the genuine Lagrangian defined on anti de Sitter spacetime. Our derivation uses the triplet formulation originating from the first-quantized BRST approach, where the action principle is determined by the BRST operator and the inner product of a first-quantised system. In this way we build, in a manifestly $so(2,d)$-covariant manner, a unifying action principle for the three types of fields mentioned above. In particular, our derivation justifies the form of some actions proposed earlier for massive and massless fields on AdS. We also give a general setup for ambient Lagrangians in terms of the respective jet-bundles and variational bi-complexes. In particular we introduce a suitable ambient-space Euler-Lagrange differential which allows to derive the equation of motion ambiently, {\it i.e.} without the need to explicitly derive the respective spacetime Lagrangian.
\end{abstract}

\newpage

\section{Introduction}

In the present work we revisit the problem of constructing an action principle for totally symmetric bosonic massive and (partially) massless fields of arbitrary spin over (anti)-de Sitter ((A)dS) spacetimes. The problem of constructing action principles for higher-spin fields dates back to the old works on massive fields in four-dimensional Minkowski spacetime by Fierz and Pauli \cite{Fierz:1939ix} (for spin $s=2$) and by Singh and Hagen \cite{Singh:1974qz,Singh:1974rc} (for $s>2$). 
The massless limit was obtained by Fronsdal and Fang \cite{Fronsdal:1978rb,Fang:1978wz} who observed the decoupling of degrees of freedom due to emergence of gauge symmetries. The formalism adopted in the above seminal works, nowadays referred to as ``metric-like'' approach, represents off-shell fields by totally-symmetric spacetime (spinor-)tensors. There is another formulation, called ``frame-like'', that is regularly used for gauge systems and where the off-shell gauge fields are represented by differential forms valued in some representations of the Lorentz algebra \cite{Vasiliev:1980as,Lopatin:1987hz}. Turning to free bosonic systems around (anti)-de Sitter ((A)dS) spacetimes, quadratic Lagrangians for totally-symmetric fields in any spacetime dimension have been obtained  both in the metric-like formulation (for massless \cite{Fronsdal:1978vb}, partially-massles \cite{Hallowell:2005np} and massive \cite{Zinov_MassAdS,Metsaev:2003cu,Hallowell:2005np} fields on (A)dS)  and in the frame-like formulation (again for massless \cite{Lopatin:1987hz}, partially-massless \cite{Skvortsov:2006at} and massive fields \cite{Zinoviev:2008ze,Ponomarev:2010st}).
\vskip 4 pt

The standard way to linearly realise (A)dS isometries is to see $(d+1)$-dimensional (anti)-de Sitter spacetime $(A)dS_{d+1}$ as a hyperboloid in a flat space $\mathbb{R}^{2,d}$ with one extra 
dimension: this is called the ``ambient'' (or ``embedding'') approach. The application of this construction to field equations for (A)dS fields dates back to Dirac \cite{Dirac:1935zz}. Its application for constructing action principles was proposed in \cite{Biswas:2002nk} where this approach 
was dubbed ``radial dimensional reduction'' by analogy with the similar method for deriving actions/equations for massive fields from their flat counterpart for massless fields in one extra dimension. This method became a standard tool for discussing linear field equations (see, {\it e.g.}, \cite{Fronsdal:1978vb,Metsaev:1995re,Metsaev:1997nj,Boulanger:2008kw} and refs therein). It has also been used for calculating quadratic actions \cite{Hallowell:2005np} for totally-symmetric tensor fields. In the case of curved ambient space, the Fefferman-Graham construction \cite{Fefferman_Graham} of Einstein metrics can be seen as a radial dimensional reduction. Let us also mention other applications of dimensional reduction in the more general context of curved spacetimes \cite{Scherk_Schwarz_79}.
\vskip 4 pt

In the present work, we restrict our analysis to $AdS_{d+1}$ bosonic systems of totally-symmetric tensor fields in the metric-like formulation on the flat ambient space $\mathbb{R}^{2,d}$. We start with $AdS_{d+1}$ equations of motion formulated in the ambient space $\mathbb{R}^{2,d}$ and aim at constructing a Lagrangian, which leads to these equations on ambient space upon variation. We make advantage of the BRST formulation of the free gauge equations of motion \cite{Barnich:2006pc,Alkalaev:2009vm,Alkalaev:2011zv}, such that the off-shell field multiplet $\Phi$ is subject to off-shell constraints ({\it e.g.}, tracelessness 
and radial differential conditions)
\begin{equation}\label{eq:BRST_description_1}
    R_{\alpha} \Phi = 0\,,
\end{equation}
and gauge symmetry
\begin{equation}\label{eq:BRST_description_2}
    \Phi \sim \Phi + \Omega \xi\,,
\end{equation}
where $\Omega$ is a Grassmann-odd nilpotent BRST operator. 
The constraints $\{R_{\alpha}\}$ together with $\Omega$ form an involutive system with 
respect to the supercommutator.
Gauge-invariant equations of motion read
\begin{equation}\label{eq:BRST_description_3}
    \Omega\Phi = 0\,,
\end{equation}
where gauge-invariance is satisfied by nilpotency of $\Omega$. The approach in question, referred to as the first-quantised BRST formulation, originates from the Batalin-Fradkin-Vilkovisky (BFV) formalism \cite{Fradkin:1977hw,Batalin:1977pb,Fradkin:1977xi}, the Batalin-Vilkovisky (BV) formalism \cite{Batalin:1981jr,Batalin:1983ggl} extension of the Becchi-Rouet-Stora-Tyutin (BRST) formalism \cite{Becchi:1974md,Becchi:1975nq,Tyutin:1975qk} and of String Field Theory (SFT) \cite{Neveu:1985ya,Neveu:1985cx,Banks:1985ff,Siegel:1985tw,Itoh:1985bb}. In the eighties a compact off-shell description of massless higher spin gauge fields in flat Minkowski spacetime of arbitrary dimension was provided in \cite{Bengtsson:1986ys,Ouvry:1986dv,Henneaux:1987cp}, a formulation that is nowadays called ``triplet''. In fact, in the latter references, the BRST operator entering the action for massless fields can be identified, after truncation, with an appropriate tensionless limit of the bosonic open string field BRST operator. 
In such a limit, there is no critical dimension and as a result those quadratic actions \cite{Bengtsson:1986ys,Ouvry:1986dv,Henneaux:1987cp}
are consistent in Minkowski spacetime of arbitrary dimension.
\vskip 4 pt

The aforementioned string field-like BRST operator \cite{Bengtsson:1986ys,Ouvry:1986dv,Henneaux:1987cp} was extended to (A)dS in \cite{Bengtsson:1990un,Buchbinder:2001bs,Bonelli:2003zu,Sagnotti:2003qa,Barnich:2006pc,BuchKrykhLavr,Reshetnyak_2023}. We refer to \cite{Fotopoulos:2006ci} for an analysis of the corresponding quadratic actions for both massive and massless totally-symmetric fields, whereas the review \cite{Fotopoulos:2008ka} also discusses the endeavours to introduce interactions in first-quantized BRST approach to higher-spin fields (see also \cite{Buchbinder_Reshetyiak-2021,Buchbinder_Reshetnyak-2022} and references therein).
\vskip 4 pt

In the present work we unify previous first-quantized BRST approaches and describe, ambiently and in a manifestly $SO(2,d)$-covariant way, the Lagrangian formalism for massive, massless and partially-massless totally-symmetric fields in $(A)dS_{d+1}$ spacetime.
In the case of partially-massless  fields, these ambient triplet  Lagrangians are new while we reproduce known results (mentioned above) for massless and massive cases.
Our main tool is the manifestly local version of the BRST first quantized approach developed in  
\cite{Barnich:2003wj,Barnich:2004cr,Barnich:2006pc,Grigoriev:2010ic} which allows to construct the jet-space Batalin-Vilkovisky  formulation as well as its equations of motion counterpart in terms of the BFV-BRST first-quantized system. This analysis is  motivated by the remarkable relation between the action for bosonic string field theory on the one hand and the BRST operator of the first-quantized string on the other hand \cite{Thorn:1988hm,Bochicchio:1986zj} as well as the analogous relation for quantized spinning particles and  their associated gauge fields~\cite{Bengtsson:1986ys,Ouvry:1986dv,Henneaux:1987cp}.

\vskip 4 pt

As a starting point, we recall how the  $(A)dS_{d+1}$ equations of motion can be obtained
via radial dimensional reduction of \eqref{eq:BRST_description_1}, \eqref{eq:BRST_description_2}, \eqref{eq:BRST_description_3} \cite{Barnich:2006pc,Alkalaev:2009vm,Alkalaev:2011zv}, and extend the latter correspondence to the Lagrangian level. We describe a generating procedure which allows one to obtain Lagrangians (which are $(d+1)$-forms on $(A)dS_{d+1}$) for totally-symmetric higher-spin fields in $(A)dS_{d+1}$ via radial dimensional reduction of particularly defined {\it ambient Lagrangians}, which are $(d+1)$-forms on $\mathbb{R}^{2,d}$. More in detail, we construct ambient Lagrangians in the form $\Lx[\Phi] = \big(\Phi,\Omega\Phi\big)$, where the non-degenerate inner product $\big(\cdot,\cdot\big)$ is valued in $(d+1)$-forms on $\mathbb{R}^{2,d}$, and is BRST-anti-invariant in the sense that 
\begin{equation}
    \big(\Psi,\hO\Phi\big) - (-)^{\gh(\Psi)} \big(\hO\Psi,\Phi\big) = \dd \mathcal{J}\,,
\end{equation}
with $\mathcal{J}$ being some $d$-form. The latter property is necessary for the proposed ambient Lagrangians to provide the correct equations of motion $\Omega\Phi = 0$ by applying formal variation:
\begin{equation}\label{eq:formal_variation}
    \delta \Lx[\Phi] = \big(\delta \Phi,\hO\Phi) + \big(\Phi,\hO \,\delta\Phi) = 2\,\big(\delta\Phi,\hO\Phi\big) + \dd\mathcal{J}\,.
\end{equation}

The corresponding variational calculus in purely ambient terms has however not been addressed in a systematic manner previously -- to the best of our knowledge (see however \cite{Joung_Taronna_2011,Joung_Taronna_2012,Joung_Taronna_2012_1}). This is a subtle issue because one should carefully take into account radial (differential) constraints on the various fields upon variation. In fact, the equivalence of the variational calculus before and after restriction to (A)dS must be rigorously demonstrated. In other words, the radial dimensional reduction of the ambient Euler-Lagrange variations must be shown to coincide with the Euler-Lagrange equations of the pullback of the ambient Lagrangian to (A)dS. One of the main goals of this paper is to present a formalism where this is ensured by construction, thus making legitimate the formal variation \eqref{eq:formal_variation}. For this purpose we utilise the formalism of jet bundles and the associated variational bicomplex, which we adapt to Dirac's ambient geometry for field-theoretical systems in $(A)dS_{d+1}$. The proposed ambient variational formalism is formulated in a coordinate-free fashion, nevertheless it provides concrete formulae to perform computations explicitly. Furthermore, the formalism can be applied to a broader class of ``ambient spaces'' than the simplest case of $\mathbb{R}^{2,d}$ equipped with a flat metric. Namely, as a suitable ambient space $\hM$ (with $\dim \hM = n+1$) one can take any trivial line bundle (thus possessing a nowhere-vanishing fundamental vector field $T$ and a closed $1$-form $\vartheta$ such that $\vartheta(T) = \mathrm{const}$) and a volume $(n+1)$-form $\vol$ (\footnote{Note that this data essentially corresponds to a Carrollian measured space together with a flat Ehresmann connection \cite[Section 3.4]{Bekaert_Oblak}}). The role of the (A)dS space is taken by a chosen section of $\hM$. In the context of the flat ambient space $\mathbb{R}^{2,d}$ in question, one has $T = r\partial_{r}$ (homogeneity operator along the radial direction), $\vartheta = -r^{-1}\,\dd r$, and the embedding $AdS_{d+1}\hookrightarrow\mathbb{R}^{2,d}$ is due to fixing $r = \ell$. 
\vskip 4 pt

The proposed jet-bundle formalism is also applicable to the Lagrangian description of massive higher-spin fields of arbitrary symmetry type via flat dimensional reduction. In this respect, let us mention that the Lagrangians obtained in \cite{Chekmenev:2019ayr} (see also earlier works \cite{Metsaev:2012uy,Metsaev:2017cuz}) are pullbacks of the ones given by the formalism described in the current work.
\vskip 4 pt

The plan of the paper is as follows. In Section \bref{fieldsambient}, totally-symmetric tensor fields on flat ambient space of dimension $d+2$ are introduced and the corresponding BRST triplet formulation of action and equations of motion for massless fields on flat 
spacetime is reviewed. The BRST extension of the constraints (such as homogeneity, tangentiality, 
tracelessness, ...) that they must obey to describe proper fields on (A)dS is also reviewed. In Section \bref{tripletLagrAdS}, we construct triplet Lagrangians for massless, partially-massless and massive  totally-symmetric tensor fields on (A)dS spacetime, as summarised in our Theorems \bref{thm:massless} and \bref{thm:massive}. The problem of constructing such (A)dS Lagrangians from the ones in flat ambient spacetime is effectively reduced to the problem of finding a BRST-invariant non-degenerate inner product satisfying some extra conditions, 
a technical problem that is solved in Appendix \bref{sec:proofs}. In Section \bref{sec:jets}, we make use of the jet-bundle formalism to construct the appropriate variational principle which justifies derivation of correct equations of motion from the previously obtained Lagrangians.
\vskip 4 pt

Various technical matters have been placed in the appendices. The appendix \bref{sec:coordinates} reviews, in a self-contained manner, 
the explicit relation between the flat ambient derivative (where $SO(2,d)$ covariance is manifest) and the intrinsic $(A)dS_{d+1}$ 
covariant derivative, without radial dependence and where only Lorentz $SO(1,d)$ covariance is manifest. This is used in Appendix \bref{sec:formulae} to derive the radial decomposition of important operators such as the Killing derivative ({\it i.e.} the symmetrised covariant derivative), the divergence and the Laplacian, together with other ingredients necessary for explicitly reformulating ambient expressions in intrinsic (A)dS terms. The technical proofs of the lemmas and theorems in the body of the paper have been placed in Appendix \bref{sec:proofs}. Finally, the coefficients \eqref{eq:coefficients_nu} appearing in the Lagrangians are related to Euler hypergeometric functions and Appell series in Appendix \bref{hypergeom}.

\section{Fields and equations of motion\\in the flat ambient space}\label{fieldsambient}

\subsection{Ambient space tensor fields and constraints}\label{sec:fields_constraints}

Consider the {\it ambient space} without the origin $\mathbb{R}^{2,d}\backslash \{0\}$ endowed with a flat metric $\eta$ carrying signature $(-,-,+,\dots,+)$ and thus introducing the notion of $SO(2,d)$ invariance. Let this space be naturally parametrised by Cartesian coordinates $X^A$ (with $A=0^{\prime},0,1\ldots, d$) such that the metric is diagonal $\eta_{AB} = \mathrm{diag}(-1,-1,+1,\dots,+1)$ and $SO(2,d)$-action is realised linearly. Under the $SO(2,d)$-action, the space $\mathbb{R}^{2,d}\backslash\{0\}$ is foliated by homogeneous subspaces, among which there are anti-de Sitter spaces $AdS_{d+1}$ introduced through a one-parameter family of natural embeddings (as one-sheeted hyperboloids), with  parameter $r > 0$ given as
\begin{equation}\label{eq:embedding}
    \eta_{BC}X^B X^C = -r^2\,.
\end{equation}
In the sequel, we denote $\mathbb{R}^{2,d}_{+}\subset \mathbb{R}^{2,d}\backslash\{0\}$ the domain which corresponds to $r\in\mathbb{R}_{+}$ (\footnote{Note that all the results from this paper equally apply to de Sitter spaces $dS_{d+1}$ via a mere change of signature of the ambient space to $\mathbb{R}^{1,d+1}$.}).
\vskip 4 pt

AdS fields are described with the aid of totally-symmetric ambient tensors fields $\Phi\left(X\middle| P\right)$ supported locally on $\mathbb{R}^{2,d}_{+}$ and presented in the form of a formal power series with tensor indices being contracted with auxiliary variables $P_A$:
\begin{equation}\label{eq:field}
    \Phi\left(X\middle| P\right) = \sum_{s = 0}^{\infty}\Phi^{(s)}\left(X\middle| P\right)\quad\text{where}\quad \Phi^{(s)}\left(X\middle| P\right):=\frac{1}{s!}\,\Phi^{A(s)}(X)\,P_{A(s)}.
\end{equation}
Here and in the sequel we adopt the following shorthand notations for symmetrised indices and for powers of auxiliary variables:
\begin{equation}
    \Phi^{A(m)} := \Phi^{A_1\dots A_m}\quad \text{and} \quad P_{A(m)} := P_{A_1}\dots P_{A_m}\,.
\end{equation}

AdS scalar fields are in one-to-one correspondence with scalar fields $\Phi(X)$ on $\mathbb{R}^{2,d}_{+}$ with a certain degree of homogeneity $-\Delta$:
\begin{equation}\label{eq:constraint_homogeneity}
    \left(X\cdot \partial_X + \Delta\right)\Phi = 0\,.
\end{equation}
We also refer to $\Delta$ as the {\it radial weight}. For an analogous one-to-one correspondence to hold for rank-$s$ AdS tensor fields, the above homogeneity constraint is to be supplemented by: $P$-homogeneity constraint
\begin{equation}\label{eq:spin-s}
    (P\cdot\partial_P - s)\,\Phi^{(s)} = 0
\end{equation}
and tangency condition
\begin{equation}\label{eq:constraints_transverse}
    X\cdot \partial_P\,\Phi = 0\,.
\end{equation}
The above constraints can be viewed as equations defining a unique lift of a field from the surface $AdS_{d+1}$ (defined by \eqref{eq:embedding}) to $\mathbb{R}^{2,d}_{+}$. Other way around, any ambient field subject to \eqref{eq:constraint_homogeneity}, \eqref{eq:spin-s} and \eqref{eq:constraints_transverse} gives an AdS tensor field upon restriction to the subspace $AdS_{d+1}\hookrightarrow\mathbb{R}^{2,d}_{+}$.

A collection of AdS tensors of ranks $s,\dots, s-t+1$ with $t\geqslant 1$ can be described ambiently by relaxing the constraint \eqref{eq:constraints_transverse} to
\begin{equation}\label{eq:constraints_transverse_depth}
    (X\cdot \partial_P)^{t}\,\Phi = 0\,.
\end{equation}
As a particular example of interest, totally-symmetric partially massless higher-spin field of spin $s$ and depth $t$ ($1\leqslant t\leqslant s$) corresponds to the critical value of the radial weight $\Delta_{s,t} = 1+t-s$ (the case $\Delta_{s} = 2-s$ with $t = 1$ corresponds to massless fields).
\vskip 4 pt

In addition to the constraints listed above, the following ambient constraints are imposed to describe irreducible spin-$s$ dynamics on $AdS_{d+1}$: the algebraic trace constraint
\begin{equation}\label{eq:constraints_trace}
    \partial_P\cdot \partial_P\,\Phi = 0\,,
\end{equation}
as well as differential constraints (denote $\Box:=\partial_X\cdot\partial_X$)
\begin{equation}\label{eq:constraints_eoms}
    \Box \Phi = 0\,,\quad \partial_{P}\cdot\partial_{X} \Phi = 0\,.
\end{equation}
Finally, for the critical values $\Delta_{s,t}$ we introduce the following gauge equivalence relation, which is compatible with the above constraints:
\begin{equation}\label{eq:gauge_gradient}
    \Phi^{(s)} \sim \Phi^{(s)} + P\cdot\partial_X\varepsilon^{(s-1)}
\end{equation}
with arbitrary $\varepsilon^{(s-1)}$ subject to constraints  \eqref{eq:constraint_homogeneity} (with $\Delta - 1$ instead of $\Delta$) and \eqref{eq:spin-s} (with $s-1$ instead of $s$).

\subsection{BRST formulation of equations of motion}

Consider for all $n=0,1,2,\dots$ the fibers of the $n$-fold symmetric tensor power ${\odot}^{n}T\mathbb{R}^{2,d}_{+}$ of the tangent bundle $T\mathbb{R}^{2,d}_{+}$ with local coordinates $X^A$ and $y^{A(n)}$ (by definition, ${\odot}^{0}T\mathbb{R}^{2,d}_{+}$ is the line bundle over $\mathbb{R}^{2,d}_{+}$). Let $\Lambda$ be a trivial bundle over $\mathbb{R}^{2,d}_{+}$ with fibers presented by Grassmann algebra over fermionic (ghost) generators $c_0,b,c$. Then we introduce the bundle
\begin{equation}\label{eq:field_bundle}
   \F = \Lambda\otimes\,{\odot} T\mathbb{R}^{2,d}_{+}\,,\quad \text{where}\quad {\odot} T\mathbb{R}^{2,d}_{+} = \bigoplus_{n = 0}^{\infty} {\odot}^{n}T\mathbb{R}^{2,d}_{+}\,,
\end{equation}
with fibers being $\mathbb{Z}$-graded spaces. The grading will be referred to as ghost-degree $\gh(\cdot)$ and fixed by setting $\gh(P_A) = 0$, $\gh(c_0) = \gh(c) = -\gh(b) = 1$. There is also an induced Grassmann $\mathbb{Z}_2$-grading $|\cdot| = \gh(\cdot)\; (\mathrm{mod}\;2)$ which conforms to the case of integer-spin fields (the latter being our only focus for the present work).
\vskip 4 pt

Ghost-extended (ambient) fields $\Phi$ are elements of the space of sections
\begin{equation}
    \hS:=\Gamma\left(\F\right)\,.
\end{equation}
In words, ghost-extended fields can be understood as functions $\Phi(P,c_0,b,c|X)$ which decompose as formal power series over monomials in $P^A$, $c_0$ and $b,c$. Components carrying a particular ghost degree admit gauge-theory interpretation: ghost-degree-$(-1)$ component contains gauge parameters,
\begin{equation}
    \left.\Phi\right|_{\mathrm{gh} = -1} = b\,\varepsilon\,;
\end{equation}
ghost-degree-$0$ component parametrises the field content of the theory,
\begin{equation}
    \left.\Phi\right|_{\mathrm{gh} = 0} = B + c_0 b\,C + cb\,D\,;
\end{equation}
ghost-degree-$1$ component is associated with equations of motion, while ghost-degree-$2$ (respectively, higher-ghost degrees) are associated with Noether identities (respectively, higher Noether identities).
Note that the sector of fields (the ghost-degree-$0$ sector) is constituted by three components $B,C,D$: this is why the description in question is referred to as {\it triplet} in the literature  (see \cite{Francia_Sagnotti_2002,Sagnotti:2003qa} and references therein). Originally this type of formulation was obtained for massless higher-spin fields in flat space by taking a suitable tensionless limit of the free bosonic string \cite{Bengtsson:1986ys,Ouvry:1986dv,Henneaux:1987cp}. At the level of equations of motion massive and (partially) massless fields in AdS are obtained by radial reduction from the flat ambient space \cite{Biswas:2002nk,Bonelli:2003zu,Barnich:2006pc}. For this reason the BRST description in question is referred to as ``triplet'' here. 
\vskip 4 pt 

In what follows, we use the term {\it operator} for those endomorphisms of the space of fields $\hS$ which are differential operators on the base performing linear transformation of the fibers. An operator $\op$ is said to carry ghost degree $g$, $\gh(\op) = g$, if for any $\Phi\in\hS$ with $\gh(\Phi) = m$ holds $\gh(\op\Phi) = m + g$. Ghost degree of an operator induces its Grassmann parity $|\op| = g \;(\mathrm{mod}\;2)$.
\vskip 4 pt

As a next step, we introduce a point-wise (with respect to $\mathbb{R}^{2,d}_{+}$) graded-symmetric inner product $\langle\cdot,\cdot\rangle$ on $\hS$. First, consider the Fock pairing $\langle \cdot,\cdot \rangle_{\mathrm{Fock}}$ on polynomials over $P^A$ extended to the ghost variables $c,b$, such that it is graded-symmetric, and the only non-trivial pairings among the generators are\footnote{An equivalent way of defining $\langle\cdot,\cdot\rangle_{\mathrm{Fock}}$ used in the literature is by considering pairs of oscillators with non-trivial graded commutation relations $[\bar{P}^A,P_B] = \delta^A_B$, $[\bar{c},b] = 1$, $[\bar{b},c] = -1$ and constructing the corresponding Fock space by setting $\bar{b}\left|0\right> = 0$, $\bar{c}\left|0\right> = 0$, $\bar{P}^A\left|0\right> = 0$, where the vacuum vector $\left|0\right>$ satisfies $\left<0\middle|0\right>=1$, $\gh(\left|0\right>) = 0$ and $\left|0\right>^{\dagger} = \left<0\right|$.}
\begin{equation}
    \langle P^A,P^B \rangle_{\mathrm{Fock}} = \eta^{AB}\,,\quad \langle c,b \rangle_{\mathrm{Fock}} = 1\,,\quad \langle 1,1 \rangle_{\mathrm{Fock}} = 1\,.
\end{equation}
The above relations fix the conjugation rules for the generators, which are extended to monomials by accepting the convention $\langle \Psi,\mathcal{O}\Phi\rangle_{\mathrm{Fock}} = (-)^{|\mathcal{O}|\cdot|\Psi|}\langle \mathcal{O}^{\dagger}\Psi,\Phi\rangle_{\mathrm{Fock}}$. Equivalently, the conjugation acts as an involutive anti-automorphism $(AB)^{\dagger} = (-)^{|A|\cdot|B|} B^{\dagger}A^{\dagger}$.
\vskip 4 pt 

The so-defined pairing is extended to take values in the functions of $X^A$ and $c_0$ by the graded-symmetric property and linearity in the left slot, namely $\langle c_0 \Psi,\Phi\rangle_{\mathrm{Fock}} = c_0\,\langle \Psi,\Phi\rangle_{\mathrm{Fock}}$ and $\langle \Psi,c_0\Phi\rangle_{\mathrm{Fock}} = (-)^{|\Psi|}c_{0}\,\langle \Psi,\Phi\rangle_{\mathrm{Fock}}$. Conjugation of generators reads as
\begin{equation}\label{eq:dagger}
    (X^A)^{\dagger} = X^A,\quad P_A^{\dagger} = \frac{\partial}{\partial P^A},\quad c_0^{\dagger} = c_0,\quad c^{\dagger} = \frac{\partial}{\partial b},\quad b^{\dagger} = -\frac{\partial}{\partial c}\,.
\end{equation}
(here and in what follows ambient indices are raised and lowered by the ambient metric). With $\langle\cdot,\cdot\rangle_{\mathrm{Fock}}$ at hand we introduce
\begin{equation}\label{eq:pairing_extended}
    \langle \cdot,\cdot\rangle = \int \dd c_0\,\,\langle \cdot,\cdot\rangle_{\mathrm{Fock}}\,,
\end{equation}
where Berezin integration is normalised as $\int \dd c_0\,c_0 = 1$. The inner product \eqref{eq:pairing_extended} is graded-symmetric and non-degenerate: {\it i)} for all $\Psi,\Phi\in\hS$ holds $\langle\Psi,\Phi\rangle = (-)^{|\Psi||\Phi|}\langle\Phi,\Psi\rangle$, and {\it ii)} if $\langle\Phi,\Psi \rangle = 0$ for all $\Psi\in \hS$, then $\Phi = 0$. Note that $\gh\big(\langle\cdot,\cdot\rangle\big) = -1$, so if $\langle\Psi,\Phi\rangle \neq 0$, then $|\Psi| = |\Phi|\pm 1$ and therefore $\langle\Psi,\Phi\rangle = \langle\Phi,\Psi\rangle$. Conjugation $(\cdot)^{\dagger}$ is extended to \eqref{eq:pairing_extended} by the following convention:
\begin{equation}
    \langle \Psi,\mathcal{O}\Phi\rangle = (-)^{|\mathcal{O}|\cdot|\Psi|}\langle \mathcal{O}^{\dagger}\Psi,\Phi\rangle + \frac{\partial}{\partial X^A}\,J^A_{\Psi,\Phi}\,,
\end{equation}
which implies, in addition to \eqref{eq:dagger}, that
\begin{equation}\label{eq:dagger_ext}
    \left(\frac{\partial}{\partial X^A}\right)^{\dagger} = -\frac{\partial}{\partial X^A},\quad \left(\frac{\partial}{\partial c_0}\right)^{\dagger} = -\frac{\partial}{\partial c_0}\,.
\end{equation}
\vskip 4 pt

The following nilpotent BRST operator $\hO:\hS\to\hS$ furnishes the triplet description of totally-symmetric higher-spin fields in the flat ambient space (see, {\it e.g.}, \cite{Barnich:2004cr} and references therein):
\begin{equation}\label{eq:BRST_operator}
    \hO = c_0\,\Box + c\,S + S^{\dagger}\,\frac{\partial}{\partial b} + c \,\frac{\partial}{\partial b}\frac{\partial}{\partial c_0},\quad S = \partial_P\cdot\partial_X,\quad S^{\dagger} = -P\cdot\partial_X\,,
\end{equation}
with $\gh(\hO) = 1$. Gauge-invariant equations of motion are $\hO\Phi = 0$ for $\gh(\Phi) = 0$, with gauge transformations $\Phi \sim \Phi + \hO\xi$ for any $\xi\in\hS$ with $\gh(\xi) = -1$.
\vskip 4 pt

Radial reduction is performed by imposing BRST-invariant extensions of the constraints \eqref{eq:constraint_homogeneity}, \eqref{eq:constraints_transverse_depth} and \eqref{eq:constraints_trace}, and leads to the equations of motion for higher-spin fields on $AdS_{d+1}$ \cite{Alkalaev:2011zv}. The BRST completion of the homogeneity operator $X\cdot\partial_X$ reads
\begin{equation}\label{eq:homogeneity_ghosts}
    h = - X\cdot \partial_X -2\, c_0\frac{\partial}{\partial c_0} + b\,\frac{\partial}{\partial b} - c\,\frac{\partial}{\partial c}\,,
\end{equation}
The pairing \eqref{eq:pairing_extended} carries $h$-degree $2$ in the sense that
\begin{equation}\label{eq:braket_h-degree}
    h\,\langle\cdot,\cdot\rangle = \langle h\,\cdot,\cdot\rangle + \langle\cdot,h\,\cdot\rangle + 2\,\langle\cdot,\cdot\rangle\,.
\end{equation}
Note that because $h^{\dagger} = -h$, the ambient operator $\op$  carries the same $h$-degree as $\op^{\dagger}$. This is verified by conjugating the relation $[h,\op] = \alpha\,\op$ (with $\alpha\in\mathbb{R}$) which leads to $[h,\op^{\dagger}] = \alpha\,\op^{\dagger}$.
\vskip 4 pt

Along the same lines, the BRST completion of the operator $P\cdot \frac{\partial}{\partial P}$ reads:
\begin{equation}\label{eq:helicity_ghosts}
    N = P\cdot\partial_{P} + b\frac{\partial}{\partial b} + c\frac{\partial}{\partial c}\,.
\end{equation}
For the bundle \eqref{eq:field_bundle} one can define its sub-bundle $\F^{(s)}\subset \F$ such that the operator \eqref{eq:helicity_ghosts} takes the value $s$ on its sections: the fibers of $\F^{(s)}$ are parametrised by the coordinates $\theta^{(s)}_{k} u^{A(k)}$ such that the monomials $\theta^{(s)}_{k}$ in the ghost generators $c_0,c,b$ satisfy $\big(b\frac{\partial}{\partial b} + c\frac{\partial}{\partial c} - (s-k)\big)\,\theta^{(s)}_{k} = 0$ (which implies, in turn, $k\in\{s,s-1,s-2\}$). For example, for $s = 2$ the fibers of $\F^{(2)}$ are parametrised by the coordinates $\theta^{(2)}_{2}u^{A(2)}$, $\theta^{(2)}_{1}u^{A}$ and $\theta^{(2)}_{0} u$ with $\theta^{(2)}_{2}\in\{1,c_0\}$, $\theta^{(2)}_{1}\in\{b,c,c_0b,c_0c\}$ and $\theta^{(2)}_{0}\in\{cb,c_0cb\}$. We denote the corresponding space of sections as
\begin{equation}\label{eq:sections_s}
    \hS^{(s)} := \Gamma(\F^{(s)})\,.
\end{equation}
Since the operators \eqref{eq:homogeneity_ghosts} and \eqref{eq:helicity_ghosts} commute, one can define the following subspaces:
\begin{equation}\label{eq:subspaces}
    \hS_{\Delta} = \left\{\Phi\in\hS\;:\;(h - \Delta)\,\Phi = 0\right\},\quad \hS^{(s)}_{\Delta} = \hS_{\Delta}\cap\hS^{(s)}\,.
\end{equation}
The BRST operator preserves $\hS^{(s)}_{\Delta}$ for any values $s = 0,1,2,\dots$ and $\Delta\in\mathbb{R}$. 
\vskip 4 pt

The ghost-extensions of tangent \eqref{eq:constraints_transverse} and trace \eqref{eq:constraints_trace} constraints are respectively:
\begin{equation}\label{eq:irred_constraints}
    \sfT_1 = X\cdot \partial_P +2c_0\frac{\partial}{\partial c},\quad \sfT_2 = \partial_P\cdot \partial_P + 2\,\frac{\partial}{\partial b}\frac{\partial}{\partial c}\,.
\end{equation}
For any $s = 0,1,2,\dots$ define the subspace $\hT^{(s)}_{\Delta}\subset \hS^{(s)}_{\Delta}$ whose elements are traceless: 
\begin{equation}\label{eq:ghost_traceless}
    \Phi\in\hT^{(s)}_{\Delta}\quad\Rightarrow \quad\sfT_2\Phi = 0\,.
\end{equation}
More to that, in the case of the special value $\Delta = 1 + t - s$ for some $t \in \{1,\dots, s\}$,
\begin{equation}\label{eq:ghost_tangent}
    \text{for any}\quad \Phi\in\hT^{(s)}_{1 + t - s}\quad \text{also impose}\quad (\sfT_1)^t\Phi^{(s)} = 0\,.
\end{equation}
As was mentioned in the comment after \eqref{eq:constraints_transverse_depth}, the latter situation corresponds to partially massless fields of depth $t$. Denote $\hT_{\Delta} = \bigoplus_{s = 0}^{\infty} \hT_{\Delta}^{(s)}$.
\vskip 4 pt

The following commutation relations hold for the above constraints and BRST operator:
\begin{equation}
    \def\arraystretch{1.4}
    \begin{array}{c}
        {[}\hO,\sfT_1{]} = - \frac{\partial}{\partial b}\,(h + N - 2) + c\,\sfT_2,\quad {[}\hO,\sfT_2{]} = 0\,,\\
        {[}h,\sfT_1{]} = -\sfT_1,\quad {[}h,\sfT_2{]} = 0, \quad {[}N,\sfT_1{]} = -\sfT_1,\quad {[}N,\sfT_2{]} = -2\sfT_2\,,\\
        {[}\sfT_1,\sfT_2{]} = 0\,.
    \end{array}
\end{equation}
As a consequence, $[\hO,(\sfT_{1})^t]\,\Phi = 0$ holds whenever $(h+N-t-1)\,\Phi = 0$ and $\sfT_{2}\Phi = 0$. Therefore $\hO$ preserves $\hT^{(s)}_{\Delta}$ for any $s = 0,1,2,\dots$ and $\Delta\in\mathbb{R}$. 
\vskip 4 pt

As a concluding remark about the ambient formulation in question, we describe the decomposition of ambient tensors with respect to tangent and normal components relatively to $AdS_{d+1}$, and apply it to resolve the $\sfT_1$-constraint explicitly (for the proof of the following Lemma see Appendix \bref{sec:proof_lem_1}).
\begin{lem}\label{lem:radial_decomposition}
    Let $T = T^A(X)\partial_A$ be a vector field on $\mathbb{R}^{2,d}_{+}$ such that $T^2 < 0$ everywhere. Then for any ambient field $\Phi\in\hS$, locally there is a uniquely defined decomposition:
\begin{equation}\label{eq:radial_decomposition}
    \Phi = \sum_{n\geqslant 0}\Phi_n\,,
\end{equation}
where
\begin{equation}
    \Phi_n = (T\cdot P)^n \,\Phi^{\perp}_{n}\quad\text{with}\quad
(T\cdot \partial_P) \,\Phi^{\perp}_{n} = 0\,.
\end{equation}
Or, equivalently, such that
\begin{equation}
\Big((T\cdot P)(T\cdot \partial_P) - n\,T^2\Big)\,\Phi_n = 0\,.
\end{equation}
In other words, the components $\Phi_n$ are of homogeneity degree $n$ in the auxiliary variable 
\begin{equation}
q = -\frac{1}{\sqrt{-T^2}} \,(T\cdot P)\,.
\end{equation}
\end{lem}
\vskip 4 pt

In the case $T = X^A\partial_A$ we will refer to the above decomposition as {\it radial decomposition}. The latter is restricted by the constraint $\sfT_1$ as follows. Note that
\begin{equation}\label{eq:resolution_T1}
    \sfT_1 \,=\, U \,\,(X\cdot\partial_P)\,\,U^{-1}\,,
\end{equation}
where
\begin{equation}\label{eq:operator_U}
U = 1-\frac{2}{X^2}\,(X\cdot P)\,c_0\frac{\partial}{\partial c}\,,\qquad U^{-1} = 1+\frac{2}{X^2}\,(X\cdot P)\,c_0\frac{\partial}{\partial c}\,.
\end{equation}
Then the following lemma applies (see Appendix \bref{sec:proof_lem_2} for proof).
\begin{lem}\label{lem:T1_resolve}
Consider a field $\Phi\in\hS$ and the operator $U$ defined in \eqref{eq:operator_U}. For any integer $t\geqslant 1$ the following conditions are equivalent:
\begin{itemize}
    \item[1.] $(\sfT_1)^t\Phi = 0$,
    \item[2.] $\Phi=U\widetilde{\Phi}$, where the radial decomposition of $\widetilde{\Phi}$, as defined in Lemma \bref{lem:radial_decomposition},
     satisfies $\widetilde{\Phi}_{m} = 0$ for all integer $m\geqslant t$\,, i.e.
    \begin{equation}
    \widetilde{\Phi} = \sum\limits_{m = 0}^{t-1}\widetilde{\Phi}_m\,.
    \end{equation}
\end{itemize}
\end{lem}

In other words, if one decomposes $P^A = \widetilde{P}^A + q\,\frac{X^A}{r}$ where the term $\widetilde{P}^A$ is such that $X\cdot \widetilde{P} = 0$ and $q = -\tfrac{1}{r}\,(X\cdot P)$ (see the the proof of Lemma \bref{lem:radial_decomposition}, Section \bref{sec:proof_lem_1}) then according to the above lemma $\widetilde{\Phi} = U^{-1}\Phi$ is a polynomial of degree at most $t-1$ in $q$ and hence satisfies $(\sfT_1)^{t}\Phi = 0$. This is the case for the description of (partially) massless fields in terms of fields $\hT^{(s)}_{1+t-s}$.

\section{Ambient formulation\\for triplet Lagrangians on $AdS_{d+1}$}\label{tripletLagrAdS}

\subsection{Differential forms and pairing}

Denote $\vol = \dd X^{0^{\prime}}\wedge\dd X^{0}\wedge\dots\wedge\dd X^{d}$ the ambient volume form and consider the $(d+1)$-form 
\begin{equation}\label{eq:vol_AdS}
    \vol_{X} = i_{_{X\cdot\partial}}\vol\,.
\end{equation}
The operator \eqref{eq:homogeneity_ghosts} is extended to act on the algebra of differential forms as a derivation, with $h\,\dd X^{A} = -\dd X^{A}$. Then $i_{X\cdot\partial}\vol_{X} = 0$ and $h\vol_{X} =-(d+2)\vol_{X}$. Pullback of $\mathcal{V}_{X}$ to $AdS_{d+1}$ coincides (up to a factor) with the $SO(2,d)$-invariant volume form on $AdS_{d+1}$ for any $r>0$.
\vskip 4 pt

Consider an inner product on $\hS$ with values in the space $\bigwedge^{d+1}\mathbb{R}^{2,d}_{+}$ of $(d+1)$-forms:
\begin{equation}\label{eq:pairing_in_difforms}
    \big(\cdot,\cdot\big)_{K} = \langle K\cdot,\cdot \rangle\,\,\vol_{X}\,,
\end{equation}
where $K$ is an invertible algebraic (with no $X$-derivatives) operator on $\hS$ such that $K^{\dagger} = K$. This is sufficient for the pairing $\big(\cdot,\cdot\big)_K$ to be symmetric and non-degenerate because so is the deformed pairing $\langle K\cdot,\cdot\rangle$. We also require that $K$ does not reshuffle different spin- and ghost-sectors, {\it i.e.} $[N,K]=0$ and $\gh(K) = 0$. In the sequel we apply the inner product \eqref{eq:pairing_in_difforms} to the elements from $\hS_{\Delta}$, {\it i.e.} fields with a particular radial weight. In this case we  fix the following value of the radial weight for $K$: $[h,K] = (d - 2\Delta)\, K$, which leads to $h\big(\Psi,\Phi)_K = 0$ (for all $\Psi,\Phi\in\hS_{\Delta}$). For the inner product at hand we define conjugation as 
\begin{equation}
\label{dashdagger}(\cdot)^{\#} = K^{-1}\,(\cdot)^{\dagger}\,K
\end{equation}
(we omit its dependence on $K$ for brevity). The following lemma is useful (see Appendix \bref{sec:proof_lem_3} for  proof).

\begin{lem}\label{lem:conjugation}
    For any $\Delta\in\mathbb{R}$ the inner product $\big(\cdot,\cdot\big)_K$ is non-degenerate on $\hS_{\Delta}$. If an operator $\op$ preserves $\hS_{\Delta}$,  then $\op^{\#}$ preserves $\hS_{\Delta}$ as well, and the two operators are conjugate with respect to $\big(\cdot,\cdot\big)_K$ in the sense that
    \begin{equation}\label{eq:conjugate_property}
    \big(\Psi,\op\Phi\big)_{K} = (-)^{\gh(\Psi)\cdot\gh(\op)}  \big(\op^{\#}\Psi,\Phi\big)_{K} + \dd\mathcal{J}_{\Psi,\Phi}
    \end{equation}
    for some $d$-form $\mathcal{J}_{\Psi,\Phi}$.
\end{lem}

Note that $\Delta\in\mathbb{R}$ in the assertion of the above lemma being fixed, for an arbitrary ambient operator $R$ and for all $\Phi,\Psi\in\hS_{\Delta}$ one has
\begin{equation}
    \big(R(h - \Delta)\,\Phi,\Psi\big)_K = 0\,\quad\text{and}\quad \big((h - \Delta)R\,\Phi,\Psi\big)_K = 0\,.
\end{equation}
Hence the property \eqref{eq:conjugate_property} is preserved under adding combinations of operators of the form $R(h - \Delta)$ or $(h - \Delta)R$ to $\op^{\#}$, with $R$ standing for arbitrary ambient operators.

\subsection{BRST-anti-invariant inner product\\ and Lagrangians for the massless fields}\label{sec:massless}

We are particularly interested in inner products restricted to specific subsets $\hT\subset\hS$ such that the BRST operator is symmetric, {\it i.e.}  $\Omega^{\#} = \Omega$. Equivalently, the inner products of interest are BRST-anti-invariant, which means
\begin{equation}\label{eq:BRST_anti-invar}
    \big(\Psi,\Omega\Phi\big)_{K} - (-)^{\gh(\Psi)}\big(\Omega\Psi,\Phi\big)_{K} = \dd\mathcal{J}_{\Psi,\Phi}\quad\text{for all}\quad \Psi,\Phi \in \hT\,.
\end{equation}
The BRST operator is symmetric with respect to the inner product \eqref{eq:pairing_in_difforms} on $\hS_{\Delta}$, so
\begin{equation}
    \hO^{\dagger} = \hO\quad\Rightarrow\quad\hO^{\#} = K^{-1}\hO K
\end{equation}
due to \eqref{dashdagger}.
\vskip 4 pt

As a next step we concentrate on massless spin-$s$ fields which are identified with the elements of $\hT^{(s)}_{2-s}$. Consider the simplest operator $K_0 = r^{-(d + 2s - 4)}$ where the radial coordinate $r$ is defined by \eqref{eq:embedding}. This operator satisfies all the conditions listed after the definition \eqref{eq:pairing_in_difforms} and denote $\big(\cdot,\cdot\big)_0 : = \big(\cdot,\cdot\big)_{K_0}$. Note the following relation for any $w\in\mathbb{R}$:
\begin{equation}
\begin{array}{rl}
    [\hO,r^{-w}] = &  -r^{-(w + 2)}\,w\,\left(c_0\,(w - 2(X\cdot\partial_X) - d) +\, (c\, X\cdot\partial_P)^{\dagger} - \,c \,X\cdot\partial_P\right)\\
    = & -w\,r^{-(w + 2)}\left(c_0\, (w  - d + 2h) + (c \sfT_1)^{\dagger} - c \sfT_1\right)\,.
\end{array}
\end{equation}
For the particular value $w_{s} = d + 2s - 4$ it implies, that for any $\Psi,\Phi\in\hT^{(s)}_{2-s}$ one has
\begin{equation}\label{eq:commutator_annihilator}
    \langle[\hO,K_0]\,\Psi,\Phi\rangle = 0\,.
\end{equation}
Indeed, this is because $(w_{s} - d + 2h)\Psi = 2(h-\Delta)\Psi = 0$ and by definition of $\hT^{(s)}_{2-s}$ holds $\sfT_1\Psi = \sfT_1\Phi = 0$. The latter implies, in turn, that any $\sfT_1^{\dagger}$-image is orthogonal to $\hT^{(s)}_{2-s}$ with respect to $\langle\cdot,\cdot\rangle$. All in all, relation \eqref{eq:commutator_annihilator} together with the sequence of identities
\begin{equation}
\def\arraystretch{1.5}
\begin{array}{rl}
\big(\hO^{\#}\Psi,\Phi\big)_0 =& 
\langle K_0(K_0^{-1}\hO K_0)\Psi,\Phi\rangle\,\vol_{X} \\
    =  & 
    \big(\hO\Psi,\Phi\big)_0 + 
    \langle[\hO,K_0]\,\Psi,\Phi\rangle\,\vol_{X}
\end{array}
\end{equation}
leads to the BRST-anti-invariance \eqref{eq:BRST_anti-invar} of the inner product $\big(\cdot,\cdot\big)_0$ on $\hT^{(s)}_{2-s}$, as defined in \eqref{eq:BRST_anti-invar}.
The following theorem is in order (see Appendix \bref{sec:proof_thms} for proof).
\begin{thm}\label{thm:massless}
    The pairing $\big(\cdot,\cdot\big)_{0}$ on $\hT^{(s)}_{2-s}$ is symmetric, non-degenerate and BRST-anti-invariant. For $\Phi\in \hT^{(s)}_{2-s}$ with $\gh(\Phi) = 0$ the following ambient Lagrangian
        \begin{equation}
           \Lx[\Phi] = \big(\Phi,\hO\Phi\big)_{0} = L[\Phi]\,\vol_{X}
        \end{equation}
        is gauge-invariant, and the pullback of $L[\Phi]$ to $AdS_{d+1}$ coincides (up to normalisation) with the gauge-invariant triplet Lagrangian \cite{Sagnotti:2003qa,Fotopoulos:2006ci,Fotopoulos:2008ka} for a massless spin-$s$ field.
\end{thm}

\subsection{Massive and partially massless AdS fields}\label{sec:massive}

The above construction can be adapted to a uniform description of massive and partially massless fields. Namely, the main freedom available by the construction is in the deformation operator $K$. Thus we are aiming at a completion $K = K_0 + \dots$ such that the corresponding inner product $\big(\cdot,\cdot\big)_K$ is BRST-anti-invariant on $\hT_{\Delta}$ for any $\Delta\in\mathbb{R}$. The latter implies (see \eqref{eq:commutator_annihilator}) that operators $K$ of interest should satisfy
\begin{equation}\label{eq:K_ambiguity}
   [\hO,K] =  R^{\prime}(h - \Delta) + (h - \Delta)R^{\prime\prime} + S^{\prime}\sfT_2 + \sfT_2^{\dagger}S^{\prime\prime}
\end{equation}
for some ambient operators $R^{\prime}, R^{\prime\prime}$ and $S^{\prime}, S^{\prime\prime}$.
\vskip 4 pt

The following result is available (see Appendix \bref{sec:proof_K} for proof).
\begin{lem}\label{lem:deformation}
    There exists a unique inner product $(\cdot,\cdot)_K$ on $\hT_{\Delta}$ (modulo overall rescaling) such that it is non-degenerate and BRST-anti-invariant for all $\Delta\in\mathbb{R}$. The operator $K$ is fixed up to adding operators having the structure of the right-hand-side of \eqref{eq:K_ambiguity}:
    \begin{equation}\label{eq:deformation_K_ansatz}
        K = r^{-(d-2\Delta)}\,\bigg(U^{\dagger}U^{-1}\bigg)^{-1}\,\mathfrak{D}\,\,\bigg(U^{\dagger}U^{-1}\bigg)\,,
    \end{equation}
    where $U$ was defined in \eqref{eq:operator_U}, and the operator $\mathfrak{D}$ is diagonal in the radial decomposition \eqref{eq:radial_decomposition}: for any non-zero radial component of $\Phi^{(s)}\in\hT^{(s)}_{\Delta}$ holds
    \begin{equation}\label{eq:coefficients_nu}
        \mathfrak{D}\Phi^{(s)}_{n} = \nu^{(s|n)}_{\Delta}\,\Phi^{(s)}_{n}\quad\text{with}\quad\nu^{(s|n)}_{\Delta} = \frac{\big[\frac{d}{2} + s - 2\big]_{n}}{\big[\Delta + s - 2\big]_{n}}\;{}_2F_{1}\left(\begin{matrix} -n,\;\frac{d}{2}-\Delta\\ \frac{d}{2} + s - 1 - n \end{matrix}\; ; \;-1\right)\,,
    \end{equation}
    where $[x]_n = x(x-1)\dots (x-n+1)$ denotes the falling Pochhammer symbol.
\end{lem}

Note that the description of massless fields proposed in the Section \bref{sec:massless} meets the above construction in the sense that the operator $K$ reduces to $K_0 = r^{-(d-2\Delta)}$. Indeed, recall that any $\Phi^{(s)}\in \hT^{(s)}_{2-s}$ has the form $\Phi^{(s)} = U\Phi^{(s)}_{0}$ (see \eqref{eq:radial_decomposition} and the comment thereafter). In particular $U^{\dagger}\Phi^{(s)}_{0} = (U^{\dagger})^{-1}\Phi^{(s)}_{0} = \Phi^{(s)}_{0}$ as well as $\mathfrak{D}\Phi^{(s)}_{0} = \Phi^{(s)}_{0}$. Therefore $K\Phi^{(s)} = KU\Phi^{(s)}_{0} = r^{-(d-2\Delta)}U\Phi^{(s)}_{0} = K_0\Phi^{(s)}$. 
\vskip 4 pt

Due to the uniqueness of the inner product $\big(\cdot,\cdot\big)_K$ we omit the reference to the operator $K$ and write simply $\big(\cdot,\cdot\big)$. The following theorem generalises the above Theorem \bref{thm:massless} giving a uniform Lagrangian description of totally-symmetric massive and (partially) massless AdS fields (see Appendix \bref{sec:proof_thms} for proof).

\begin{thm}\label{thm:massive}
    The pairing $\big(\cdot,\cdot\big)$ on $\hT_{\Delta}$ is symmetric, non-degenerate and BRST-anti-invariant. For $\Phi\in \hT^{(s)}_{\Delta}$ with $\gh(\Phi) = 0$ the following ambient Lagrangian
        \begin{equation}
           \Lx[\Phi] = \big(\Phi,\hO\Phi\big) = L[\Phi]\,\vol_{X}
        \end{equation}
        is gauge-invariant, and the pullback of $L[\Phi]$ to $AdS_{d+1}$ gives the gauge-invariant triplet Lagrangian for a (partially) massless (when $\Delta = 1 + t - s$ for $t\in \{1,\dots, s\}$) or massive spin-$s$ field.
\end{thm}

Let us comment that fields of the above formulations are associated to the representation space basis elements of vanishing ghost degree. By extending $\Phi$ with the fields associated to the basis elements of nonvanishing degree results automatically in the BV master actions of the system, see {\it e.g.} \cite{Barnich:2004cr} for further details.

\subsection{From radial to flat dimensional reduction}\label{sec:flat_limit}

The described above radial reduction for Lagrangian description of massive and (partially) massless fields on $AdS_{d+1}$ admits a flat limit which leads to the description of massive fields via dimensional reduction in flat space \cite{Alkalaev:2008gi,Chekmenev:2019ayr}. To do so, we fix the timelike direction $\partial_{0^{\prime}}=V^A\partial_A$ (with $V^AV_A = -1$) and consider neighbourhoods of the points $\ell \,V^A$, each neighbourhood parametrised as $X^{A} = \ell \,V^A + \bar{X}^A$. The flat limit is understood as $\ell\to \infty$.
\vskip 4 pt

In order to analyse the flat limit of constraints and ambient Lagrangians, let us first consider the case of massive fields. For each value of $\ell$ we take a particular (non-critical) value of the radial weight $\Delta_{\ell}$ such that there exists a limit:
\begin{equation}\label{eq:weight_flat_limit}
    \frac{\Delta_{\ell}}{\ell} \xrightarrow{\ell\to\infty} m \in\mathbb{R}\,.
\end{equation}

Radial constraint \eqref{eq:homogeneity_ghosts}, rewritten in the coordinates of the point $\ell V + \bar{X}$, admits a limit, which corresponds to the flat dimensional reduction:
\begin{equation}\label{eq:constraint_homogeneity_flat}
    \ell\left(\partial_{0^{\prime}} + \frac{1}{\ell}\left(Y\cdot\partial + 2 c_0\frac{\partial}{\partial c_0} - b\frac{\partial}{\partial b} + c\frac{\partial}{\partial c}\right) + \frac{\Delta_{\ell}}{\ell}\right) \Phi = 0\quad \xrightarrow{\ell\to\infty}\quad  \left(\frac{\partial}{\partial \bar{X}^{0^{\prime}}} + m\right)\Phi = 0\,.
\end{equation}

To analyse the limit of the inner product $\langle K\cdot,\cdot\rangle\, \vol$ first note that $\ell^{-1}\vol \xrightarrow{\ell\to\infty} i_{\partial_{0^{\prime}}}(*1)$, whose pullback to the surface $\bar{X}^{0^{\prime}} = 0$ gives the flat $(d+1)$-dimensional volume form. For the limiting behavior of the deformation operator $K$ first note that $U\xrightarrow{\ell\to\infty} 1$ in \eqref{eq:operator_U} because $r^{-1} \xrightarrow{\ell\to\infty} 0$. Due to the structure of the coefficients $\nu^{(s|n)}_{\Delta}$ (see Theorem \bref{lem:deformation} and Appendix \bref{hypergeom}) one finds $\nu^{(s|n)}_{\Delta} \xrightarrow{\ell\to\infty} (-1)^{n}$, and hence 
\begin{equation}\label{eq:K_flat}
    \text{for}\quad \big(\,(X\cdot P)(X\cdot \partial_P) - n\,X^2\,\big)\,\Phi = 0\quad\text{there is}\quad \ell^{d-2\Delta_{\ell}} K\Phi\xrightarrow{\ell\to\infty} (-)^{n}\Phi\,.
\end{equation}
Next, note the following limit of the radial oscillator:
\begin{equation}
    \underbrace{-r^{-1}\,X\cdot P}_{q} \xrightarrow{\ell\to\infty} P^{0^{\prime}}\,.
\end{equation}
As a result, field components with particular degree in the radial oscillator turn into the components with the same degree in $P^{0^{\prime}}$ (see Lemma \bref{lem:radial_decomposition} for $T = \partial_{0^{\prime}}$):
\begin{equation}\label{eq:radial_decompodition_flat}
    \big((X\cdot P)(X\cdot \partial_P) - n\,X^2\big)\,\Phi = 0 \quad\xrightarrow{\ell\to\infty}\quad \left(P^{0^{\prime}}\cdot\frac{\partial}{\partial P^{0^{\prime}}} - n \right)\,\Phi = 0\,.
\end{equation}
With this at hand, relation \eqref{eq:K_flat} can be interpreted as a change of the conjugation rule for the oscillator $P^{0^{\prime}}$, which can be taken into account by introducing the following inner product:
\begin{equation}\label{eq:inner_product_flat}
    \ell^{d-2\Delta_{\ell}-1}\big(\cdot,\cdot\big) \quad \xrightarrow{\ell\to\infty}\quad\big(\cdot,\cdot\big)^{\prime} = e^{2m\,\bar{X}^{0^{\prime}}}\langle \cdot,\cdot\rangle^{\prime}\,i_{\partial_{0^{\prime}}}\vol\,,
\end{equation}
where $\langle \cdot,\cdot\rangle^{\prime}$ coincides with $\langle \cdot,\cdot\rangle$ except for the conjugation rule for the oscillator along $0^{\prime}$, which is changed to $(P^{0^{\prime}})^{\dagger} = \frac{\partial}{\partial P^{0^{\prime}}}$.
\vskip 4 pt

From the flat limit for the inner product \eqref{eq:inner_product_flat}, it is now straightforward to introduce the flat limit of the ambient Lagrangians in Theorem \bref{thm:massive}:
\begin{equation}\label{eq:Lagrangian_flat}
    \Lx^{\prime}[\Phi] = \big(\Phi,\Omega\Phi\big)^{\prime}\,,\quad \Phi\in\hS^{(s)}\,,
\end{equation}
to be accompanied by the mass constraint \eqref{eq:constraint_homogeneity_flat}, as well as the trace constraint \eqref{eq:ghost_traceless} (which is unaffected by the shift $X = \ell V + \bar{X}$). Pullback of \eqref{eq:Lagrangian_flat} to the surface $\bar{X}^{0^{\prime}} = 0$ leads to the ambient Lagrangian description of massive totally-symmetric spin-$s$ field in flat space. Note that the flat limit of equations of motion in the BRST formulation was proposed in \cite{Alkalaev:2009vm}, while flat dimensional reduction in similar terms was considered in \cite{Chekmenev:2019ayr} (see also earlier works \cite{Bekaert:2003uc,Alkalaev:2008gi} and \cite{Metsaev:2012uy,Metsaev:2017cuz}).
\vskip 4 pt

Finally, for (partially) massless fields, {\it i.e.} when $\Delta_{s,t} = 1+t-s$ for some $t \in \{1,\dots,s\}$, one proceeds along the same lines as in the massive case, but this time keeping $\Delta_{\ell} = \Delta_{s,t}$ fixed. This leads to \eqref{eq:constraint_homogeneity_flat} with $m = 0$. From the structure of the radial decomposition for (partially) massless fields (see Lemma \bref{lem:T1_resolve}) together with \eqref{eq:radial_decompodition_flat}, in the flat limit one has the decomposition $\Phi^{(s)} = \sum_{n=0}^{t-1} \Phi_n^{(s-n)}$ with respect to the homogeneity degree in $P^{0^{\prime}}$. The flat ambient Lagrangian \eqref{eq:Lagrangian_flat} splits into a direct sum with 
\begin{equation}\label{eq:Lagrangian_pm_decomposition}
    \Lx^{\prime}[\Phi^{(s)}] = \sum_{n = 0}^{t-1}\Lx^{\prime}[\Phi_n^{(s-n)}]\,,\quad \Phi^{(j)}_n\in\hS^{(j)}\,,
\end{equation}
where each $n$th component in the above decomposition describes a massless spin-$(s-n)$ field. The decomposition \eqref{eq:Lagrangian_pm_decomposition} expresses the known fact that a partially massless spin-$s$ field of depth $t$ in the flat limit splits into a set of free massless fields of spins $s,s-1,\dots,s-t+1$ \cite{Deser:2001us}.

\section{Ambient variational principle}\label{sec:jets}

As constructed in the previous sections, ambient Lagrangians are $(d+1)$-forms on the $(d+2)$-dimensional ambient space in contrast to the standard setup where  Lagrangians are spacetime top forms. In this section we develop  a description of the standard variational bicomplex approach in terms of   fields on ambient space.
In particular we give an explicit definition of the Euler-Lagrange derivative in terms of the ambient one.
Finally, we explain how the above ambient Lagrangians for totally-symmetric higher spin fields on $AdS_{d+1}$ fit into this more general setup. In addition, we demonstrate that the usual flat dimensional reduction formalism can also be understood as a particular case of the ambient Lagrangian formulation.

\subsection{Jet-bundle description of constrained ambient-space fields}

\paragraph{Ambient space and fields.} As a general setup, consider an $(n+1)$-dimensional manifold $\hM$
which serves as an ambient space for its $n$-dimensional submanifold $\Sigma$ and which carries a nowhere-vanishing vector field $T$ transversal to $\Sigma$.  It follows that $\hM$ is foliated by the integral curves of $T$ and hence is locally a bundle over $\Sigma$. For simplicity, we assume that $\hM$ is globally a principal line bundle over $\Sigma$. With a slight abuse of notation, we also consider $\Sigma$ as a global section of $\hM\twoheadrightarrow\Sigma$. 
\vskip 4 pt

It is convenient to introduce an adapted coordinate system $(t,x^\mu)$ such that
\begin{equation}
T=\frac{\partial}{\partial t}\qquad Tx^\mu=0
\end{equation}
and $\Sigma$ is singled out by $t=0$.  It is clear that $x^\mu$ define a coordinate system on $\Sigma$. Similarly, one can consider submanifolds $t=t_0$. Such manifolds can be seen as covariantly constant (i.e. 
horizontal) sections of $\hM\twoheadrightarrow \Sigma$ determined by an Ehresmann connection one-form $\vartheta=-dt$.
\vskip 4 pt

Local coordinates $X^A$ on $\hM$ are called {\it homogeneous} of degree $w$ if 
\begin{equation}\label{eq:homogeneous_coordinates}
    [T,\partial_A] = -w\,\partial_A\,\quad\text{(for some $w\in\mathbb R$)}\,.
\end{equation}
Via a $t$-dependent rescaling of the adapted coordinates one can construct homogeneous coordinates \eqref{eq:homogeneous_coordinates}, proving that the latter choice is always accessible.
For example, in the previously considered case of $\hM = \mathbb{R}_{+}^{2,d}$ one has $T = X\cdot \partial_X$, so \eqref{eq:homogeneous_coordinates} holds for the  flat coordinates with $w = 1$. Furthermore
$\vartheta = -(X\cdot X)^{-1}\,X\cdot \dd X$.
\vskip 4 pt

Finally, we assume that $\hM$ is endowed with a volume form 
\begin{equation}\label{volform}
    \vol = \rho(X)\,\dd X^{0}\wedge\dots\wedge \dd X^{n}\,.
\end{equation} 
It defines a volume form on $\Sigma$ as a pullback of
\begin{equation}\label{eq:induced_volume}
    \vol_{T} := i_{T} \vol\,.
\end{equation}
In the adapted coordinates  $\vol_{T} = \rho(X(t,x))\left|\frac{\partial X}{\partial (t,x)}\right|\,\dd x^{1}\wedge \dots \wedge \dd x^{n}$.
The data $(\hM,\Sigma,T,\vol)$ will be referred to as {\it ambient space for $\Sigma$}.

\vskip 4 pt

Consider an ambient field $\Phi(X) = \{\Phi^{\sfa}(X)\}$, where the index $\sfa$ parametrises a finite set of field components. Fields can be viewed as sections of a trivial vector bundle $\E\twoheadrightarrow\hM$, with fibers parametrised by the coordinates $u = \{u^{\sfa}\}$. Evaluation of $u^{\sfa}$ on a section $\Phi:\hM\hookrightarrow\E$ leads to the usual expressions for the fields as functions $\Phi^{\sfa}(X):=(u^{\sfa}\circ\Phi)(X)$.
\vskip 4 pt

We are interested in sections which verify the following differential constraints:
\begin{equation}\label{eq:weight_constraint_general}
    (T + \Delta^{\sfa})\,\Phi^{\sfa}(X) = 0\quad\text{for some fixed}\quad \Delta^{\sfa} \in\mathbb{R}
\end{equation}
(with no summation over $\sfa$). Note that this condition is formulated using a fixed local frame of $\E$. To formulated it in a generic frame one can, {\it e.g.}, introduce a flat linear connection on $\E$. For simplicity, in this section we disregard possible algebraic constraints on $\Phi^{\sfa}$, which can be necessary in applications. Such constraints can always be solved in terms of independent components of fields.
\vskip 4 pt

Next we define the vector bundle  as the pullback bundle $E=\E|_\Sigma$ of the vector bundle $\E\twoheadrightarrow\hM$ along the section $\Sigma\hookrightarrow \hM$. We have the following:
\begin{lem}
    Sections of $\E$ satisfying \eqref{eq:weight_constraint_general} are in one-to-one correspondence with the unconstrained sections of $E$.
\end{lem}
\begin{proof}
    In the adapted coordinates $(t,x^\mu)$, the constraints \eqref{eq:weight_constraint_general} are a system of first-order ODE's in $t$-variable so that a solution with the initial data at $t=0$ exists and is unique. In other words solutions are reconstructed in the neighbourhood of $\Sigma$ from the initial data $\phi^{\sfa}(x)$, the latter being nothing else but a section of $E$.
\end{proof} 

\paragraph{Jets of constrained sections.} A standard geometric language for variational calculus is the jet-bundle formalism (see {\it e.g.} \cite{Anderson_bicomplex} for a review). In the sequel we will adapt it to the case of constrained fields in the ambient space. As a first step, one promotes $\E$ to its (infinite) jet extension $J^{\infty}\E$ (or $J\E$ for brevity) which itself is a vector bundle over $\hM$, whose adapted local coordinates are $X^A$ and $\boldsymbol{u} = \big\{ u^{\sfa}_{A(q)}\,:\, q\geqslant 0\big\}$. For any section $\Phi\in\Gamma(\E)$ there is a uniquely defined {\it (infinite) jet prolongation} $j^{\infty}\Phi\in \Gamma(J\E)$ parametrised at each point by all its partial derivatives:
\begin{equation}\label{eq:jet_coordinates}
    u^{\sfa}_{A(q)}\circ j^{\infty}\Phi  =  \partial_{A_1}\dots\partial_{A_q}\Phi^{\sfa} \quad \text{for all}\quad q\geqslant 0\,.
\end{equation}

The jet bundle $J\E \twoheadrightarrow \hM$ carries a canonical horizontal distribution spanned by all vectors which are tangent to the jet prolongation of any section of $\E$. In local coordinates, horizontal vectors are spanned by the basis
\begin{equation}\label{eq:total_vec_basis}
    D_A = \partial_A + \sum_{q = 0}^{\infty} u^{\sfa}_{AB(q)}\,\partial_{\sfa}^{B(q)}\,,\quad \text{where}\quad \partial_{\sfa}^{B(q)} = \frac{\partial}{\partial u^{\sfa}_{B(q)}}\,.
\end{equation}
Horizontal vector fields are also referred to as {\it total vector fields} (see \cite{Anderson_bicomplex} for details). 
\vskip 4 pt

Equation \eqref{eq:weight_constraint_general} is a partial differential equation (PDE), which defines (or, can be viewed as) a vector sub-bundle $\mathfrak{i}_{\Delta} : \I\hookrightarrow J\E$ singled out by the following linear equations: 
\begin{equation}\label{eq:jet_sub-bundle}
    z^{\sfa}_{A(q)} = D_{A(q)} \big(T^B u^{\sfa}_B\big) + \Delta^{\sfa} u_{A(q)} = 0\,.
\end{equation}
Indeed, the pullback of \eqref{eq:jet_sub-bundle} by the  jet prolongation $j^{\infty}\Phi$ of a section of $\E$ reproduces the left-hand-side of \eqref{eq:weight_constraint_general} and its differential consequences. Although $\mathcal{I}$ is not a jet bundle associated to any bundle, one has the following lemma (recall the bundle $E$ introduced below \eqref{eq:weight_constraint_general}).
\begin{lem}\label{lem:pullback_bundle_I}
    Let $\left.\I\right|_{\Sigma}$ be the vector bundle over $\Sigma$ defined as the pullback of the vector bundle $\I\twoheadrightarrow\hM$ along $\Sigma\hookrightarrow \hM$. This vector bundle is isomorphic to the infinite jet bundle $JE\twoheadrightarrow\Sigma$. The horizontal distribution on $\left.\I\right|_{\Sigma}$ is obtained by the horizontal lift of the tangent space of $\Sigma$.
\end{lem}
\begin{proof}
    The base manifold being the same for the two vector bundles, one needs to establish an isomorphism between the fibers of the two. Recall that a fiber of $JE$ over a point $p\in\Sigma$ is constituted by equivalence classes of sections of $E$ with coinciding derivatives up to all orders at $p$. As for $\left.\I\right|_{\Sigma}$, note that \eqref{eq:jet_sub-bundle} contains no algebraic constraints, so the $0$th-jet projection of $\left.\I\right|_{\Sigma}$ is isomorphic to $E$. In the adapted coordinates it is evident that a horizontal section of $\I$ at each point is parametrised by the values of its projection to the $0$th jets, as well as by the values of all its derivatives along $\Sigma$ (the values of $t$-derivatives are reconstructed from the equations \eqref{eq:jet_sub-bundle}). In particular, this means that horizontal distribution on $\I$ is the horizontal lift of tangent spaces of $\Sigma$. As a result, a fiber of $\left.\I\right|_{\Sigma}$ (over a point of $\Sigma$) consists of equivalence classes of sections of $E$, with two sections belonging to the same class whenever all their derivatives (along $\Sigma$) at this point coincide. 
\end{proof}

To illustrate the above lemma in more explicit terms, one makes use of the adapted coordinates $(t,x^{\mu})$, and notes that horizontal sections of $\I$ are parametrised by $t$- and $x$-derivatives of the general solution of the constraints \eqref{eq:weight_constraint_general},
\begin{equation}\label{eq:radial_dependence}
    \Phi^{\sfa}(X(t,x)) = e^{-t\Delta^{\sfa}}\,\phi^{\sfa}(x)\,, 
\end{equation}
with $\phi(x) = \{\phi^{\sfa}(x)\}$ being a field on $\Sigma$. Since $t$-derivatives of the above functions are prescribed, $X$-derivatives of $\Phi(X)$ reduce to $x$-derivatives of $\phi(x)$, which brings one to the jet bundle $JE$.
\vskip 4 pt

Consider the following vector field on $\E$:
\begin{equation}\label{eq:T_on_E}
    \bar{T} = T^A\partial_A - \Delta^{\sfa}\,u^{\sfa}\partial_{\sfa}\,.
\end{equation}
It is clearly a symmetry of  \eqref{eq:weight_constraint_general}, in the sense that \eqref{eq:T_on_E} preserves the space of solutions viewed as the submanifold 
$\I\subset J\E$.
In other words, the prolongation of \eqref{eq:T_on_E} to a vector field on $J\E$ is tangent to $\I$,  and reads as\footnote{The definition and explicit formulae for prolongations of vector fields can be found, {\it e.g.}, in  \cite
[Proposition 1.12]{Anderson_bicomplex}.}
\begin{equation}
    \mathrm{pr}\,\bar{T} = T^{A}D_{A} - Z\,,\quad\text{where} \quad Z = \sum_{n\geqslant 0} z^{\sfa}_{A(q)}\,\partial _{\sfa}^{A(q)} = \pr\left(z^{\sfa}\frac{\partial}{\partial u^{\sfa}}\right)\,,
\end{equation}
with $z^{\sfa}_{A(q)}$ defined in \eqref{eq:jet_sub-bundle}. In what follows we denote $T^{A} D_{A}$ by $T$ whenever it does not lead to a confusion.
Note that since $T$ is the horizontal lift of a vector field from the base $\hM$, one has the following properties
\begin{equation}\label{eq:properties_T_Z}
\def\arraystretch{1.4}
\begin{array}{c}
    \ddv  i_{T} + i_{T} \ddv = 0\quad\Rightarrow\quad \LD_{T} = \ddh i_{T} + i_{T} \ddh\,,\;\; \ddv \LD_{T} = \LD_{T}\ddv\quad\Rightarrow\quad \ddh \LD_{T} = \LD_{T}\ddh\,,\\
    i_{T}\LD_{Z} = \LD_{Z} i_{T}\,,\;\; \LD_{T}\LD_{Z} = \LD_{Z} \LD_{T}\,.
\end{array}
\end{equation}

\paragraph{Bicomplex of $T$-forms.}

Consider $\bigwedge J\E$ the algebra of local forms on $J\E$. One distinguishes the {\it contact $1$-forms}
\begin{equation}\label{eq:contact_forms}
    \theta^{\sfa}_{A(q)} = \dd u^{\sfa}_{A(q)} - \dd X^B\,u^{\sfa}_{BA(q)}\quad \text{(for all $q \geqslant 0$),}
\end{equation}
whose pullback to the jet prolongation $j^{\infty}\Phi$ of any section of $\E$ vanishes. Other way around, if for some section $j^{\infty}\Phi\in\Gamma(J\E)$ the pullback $(j^{\infty}\Phi)^{*}\theta^{\sfa}_{A(q)} = 0$ for all $q\geqslant 0$, then $j^{\infty}\Phi$ is the jet prolongation of some $\Phi\in \Gamma(\E)$. The subalgebra of $\bigwedge J\E$ generated by the contact $1$-forms is referred to as {\it contact ideal}. Since total vector fields, spanned by \eqref{eq:total_vec_basis}, are tangent to the jet prolongation of any section of $\E$, they are annihilated by the contact forms \eqref{eq:contact_forms}.
\vskip 4 pt

The $1$-forms $\dd X^A$ and $\theta^{\sfa}_{A(q)}$ (for  $q\geqslant 0$) generate a local frame of $\bigwedge J\E$, so that any $\alpha \in \bigwedge J\E$ can be locally written as
\begin{equation}
    \alpha = \sum_{p,q\geqslant 0}\alpha_{A_1\dots A_p\,|\,\sfa_1;\dots ;\sfa_q}^{B_1(k_1);\dots; B_n(k_q) }(X,\boldsymbol{u})\,\,\dd X^{A_1}\wedge \dots \wedge \dd X^{A_p} \wedge \theta^{\sfa_1}_{B_1(k_1)}\wedge \dots \wedge \theta^{\sfa_q}_{B_n(k_q)}\,.
\end{equation}
The space $\bigwedge J\E$ is bi-graded by the {\it horizontal} and {\it vertical} degrees (in $\dd X^A$ and $\theta^{\sfa}_{A(n)}$, respectively):
\begin{equation}
    \bigwedge J\E = \,\bigoplus\limits_{i=0}^{n+1}\,\bigoplus\limits_{j=0}^{\infty}\,
    \bigwedge{}^{(i,j)} J\E\,.
\end{equation}
In particular, the de Rham differential splits into its horizontal and vertical parts $\dd = \ddv + \ddh$ with respect to the bi-degree:
\begin{equation}
\def\arraystretch{1.5}
\begin{array}{l}
    \ddh : \bigwedge^{(i,j)} J\E \to \bigwedge^{(i+1,j)} J\E\,, \\
    \ddv : \bigwedge^{(i,j)} J\E \to \bigwedge^{(i,j+1)} J\E\,,
\end{array}
\end{equation}
which implies $\ddh\ddh = 0$, $\ddv\ddv = 0$ and $\ddh\ddv + \ddv\ddh = 0$ as a consequence of $\dd^2 = 0$. One has the following explicit formulae for $\ddh$, $\ddv$ in terms of the basis $1$-forms $\dd X^A$ and $ \theta^{\sfa}_{A(m)}$ (see \eqref{eq:total_vec_basis} and \eqref{eq:contact_forms}):
\begin{equation}
    \dd = \dd X^{A} \partial_A + \sum_{q \geqslant 0} \dd u^{\sfa}_{A(q)}\,\partial_{\sfa}^{A(q)} = \underbrace{\dd X^{A} D_A}_{\ddh} + \underbrace{\sum_{q \geqslant 0} \theta^{\sfa}_{A(q)}\,\partial_{\sfa}^{A(q)}}_{\ddv}
\end{equation}
Note that $\theta^{\sfa}_{A(q)} = \ddv u^{\sfa}_{A(q)}$, which we will use to denote contact $1$-forms in the sequel. All in all, the algebra $\bigwedge J\E$ is a bicomplex:
\begin{equation}\label{eq:bicomplex}
    \begin{array}{ccccccccc}
        \bigwedge^{(0,0)}J\E & \xrightarrow{\ddh} & \bigwedge^{(1,0)}J\E & \xrightarrow{\ddh} & \dots & \xrightarrow{\ddh} & \bigwedge^{(n+1,0)}J\E& \xrightarrow{\ddh} & 0\\ 
        \downarrow \scriptstyle{\ddv} & \hfill & \downarrow \scriptstyle{\ddv} & \hfill & \hfill & \hfill & \downarrow \scriptstyle{\ddv}& \hfill & \hfill \\
        \bigwedge^{(0,1)}J\E & \xrightarrow{\ddh} & \bigwedge^{(1,1)}J\E & \xrightarrow{\ddh} & \dots & \xrightarrow{\ddh} & \bigwedge^{(n+1,1)}J\E& \xrightarrow{\ddh} & 0\\
        \downarrow \scriptstyle{\ddv} & \hfill & \downarrow \scriptstyle{\ddv} & \hfill & \hfill & \hfill & \downarrow \scriptstyle{\ddv}& \hfill & \hfill\\
        \dots & \hfill & \dots & \hfill & \hfill & \hfill & \dots & \hfill & \hfill\\
    \end{array}
\end{equation}
(see \cite{Anderson_bicomplex} for further details).
\vskip 4 pt

The ambient space $\hM$ is a principal bundle over $\Sigma$ with $T$ as fundamental vector field. Basic forms are differential forms $\alpha$ on $\hM$ which are both invariant ($\LD_{T}\alpha = 0$) and horizontal ($i_{T}\alpha = 0$). Equivalently, they are pullbacks of differential forms on the base $\Sigma$ along the projection $\hM\twoheadrightarrow\Sigma$.  Analogously, one introduces the subalgebra of $T$-forms on $\I$:
\begin{equation}\label{eq:forms_I_equiv}
    \tensor*{\bigwedge\nolimits}{_{T}}\I =\left\{\alpha \in \bigwedge \I\;:\; i_{T}\alpha = 0\,,\quad \LD_{T}\alpha = 0\right\} \subset \bigwedge \I\,.
\end{equation}
Recall that $\I$ is an infinitely prolonged PDE \eqref{eq:weight_constraint_general} and hence the horizontal distribution in $J\E$ is tangent to $\I$. This induces  the decompoistion of $\bigwedge \I$ into the homogeneous components with respect to horizontal/vertical bi-degree, as well as the decomposition $\dd = \ddh + \ddv$. 
\begin{lem}\label{lem:bicomplex_classes}
    The subalgebra $\bigwedge_{T}\I \subset \bigwedge \I$ is preserved by $\ddh$ and $\ddv$, and hence is a bi-complex.
\end{lem}
\begin{proof}
It is sufficient to check that both $\dd$ and $\ddh$ preserve $\bigwedge_{T}\I$. For this purpose one makes use of the Cartan's formula for $\LD_{T}$ and recalls \eqref{eq:properties_T_Z}.
\end{proof}
The algebra $\bigwedge_{T}\I$ is related to the bi-complex $\bigwedge JE$ via the following lemma. Define the map $\varepsilon_{\Sigma} : \bigwedge_{T} \I \to \bigwedge JE$ by restricting a local form to $\left.\I\right|_{\Sigma}$ and recalling the isomorphism in Lemma \bref{lem:pullback_bundle_I}.
\begin{lem}\label{lem:isomorphism_epsilon}
    The map $\varepsilon_{\Sigma}$ is an isomorphism of bi-complexes. As a consequence, if $\ddh \alpha \in \bigwedge_{T}\I$, one can find $\alpha^{\prime}\in \bigwedge_{T}\I$ such that $\ddh \alpha = \ddh\alpha^{\prime}$.
\end{lem}
\begin{proof}
    The map $\varepsilon_{\Sigma}$ is a morphism of bi-complexes. Indeed, $\varepsilon_{\Sigma}$ commutes with $\dd$ as a pullback, and with $\ddv$ because fibers are unaffected upon restriction to a submanifold $\Sigma \hookrightarrow \hM$.
    \vskip 4 pt
    Let us show that the kernel of $\varepsilon_{\Sigma}$ is trivial. Indeed, the restriction of $\alpha\in\bigwedge_{T}\I$ to the submanifold $\Sigma\hookrightarrow \hM$ serves as the initial data for the condition $\LD_{T}\alpha = 0$ (the first-order PDE), which (together with the condition $i_{T}\alpha = 0$) allows one to reconstruct $\alpha$. If the initial data is zero, so is the resulting form $\alpha$. The same argument allows one to conclude that $\varepsilon_{\Delta}$ is onto.
    \vskip 4 pt
    For any form $\alpha\in\bigwedge\I$ denote $\alpha_{\Sigma}$ its pullback to $\left.\I\right|_{\Sigma}$. Then for $\ddh\alpha\in \bigwedge_{T}\I$ one has $\varepsilon_{\Sigma} \ddh \alpha = \ddh \alpha_{\Sigma}$. Using $\alpha_{\Sigma}$ as the initial data on $\Sigma$, one solves the conditions in \eqref{eq:forms_I_equiv}, and thus reconstructs $\alpha^{\prime}\in\bigwedge_{T}\I$ such that $\alpha_{\Sigma} = \varepsilon_{\Sigma}\alpha^{\prime}$. All in all one has $\varepsilon_{\Sigma}\ddh \alpha = \varepsilon_{\Sigma}\ddh \alpha^{\prime}$, and the fact that $\varepsilon_{\Sigma}$ is an isomorphism finishes the proof.
\end{proof}
\vskip 4 pt

A convenient way to work with elements of $\bigwedge \I$ consists in considering equivalence classes of local forms on $J\E$, such that two forms in the same class are mapped to the same form upon pullback to $\I$. This identification is implied in the sequel, so for any $\alpha\in J\E$ one denotes the corresponding class by $[\alpha]\in \bigwedge\I$. In order to identify $T$-forms on $\I$ in terms of equivalence classes, define the sub-algebra of {\it ambient $T$-forms} $\bigwedge_{T} J\E \subset \bigwedge J\E$ as follows:
\begin{equation}\label{eq:difforms_subspace_strict}
   \tensor{\bigwedge\nolimits}{_T} J\E = \left\{\alpha\in\bigwedge J\E\;:\; i_{T}\alpha = 0 \quad \text{and} \quad\LD_{T-Z}\alpha = 0\right\}\,. 
\end{equation}
The two conditions in the above definition use different vector fields, while consistency is assured by the fact that $T$ and $Z$ commute because $Z$ is the prolongation of an evolutionary vector field, and $T$ is the horizontal lift of a vector field on the base $\hM$ (recall \eqref{eq:T_on_E}).
\vskip 4 pt
Although the space of $T$-forms is not preserved by $\dd$, there is a way to endow $\bigwedge_T J\E \subset \bigwedge J\E$ with the structure of a bi-graded differential algebra. Recall that $\vartheta$ denotes an Ehresmann connection one-form on the line bundle $\hM\twoheadrightarrow \Sigma$. With a slight abuse of notation, we also write $\vartheta$ for its pullback to $\bigwedge J\E$. Let us define
\begin{equation}\label{eq:differential_prime}
    \dd^{\prime} = \dd + \vartheta\wedge \LD_{Z}\,.
\end{equation}
Note that the above operator admits the following decompostion with respect to the bi-degree:
\begin{equation}\label{eq:differential_prime_bidegree}
    \dd^{\prime} = \ddh^{\prime} + \ddv^{\prime}\,,\quad \text{where}\quad \ddh^{\prime} = \ddh + \vartheta\wedge \LD_Z\quad\text{and}\quad \ddv^{\prime} = \ddv\,.
\end{equation}
\begin{lem}\label{lem:d_prime}
\begin{itemize}
    \item[1)] The following analog of the Cartan's magic formula takes place in $\bigwedge J\E$:
    \begin{equation}\label{eq:Cartan_prime}
        \LD_{T-Z} = i_T \dd^{\prime} + \dd^{\prime} i_T\,.
    \end{equation}
    \item[2)] The operators $\ddh^{\prime}$ and $\ddv$ verify Leibniz rule. Moreover, the operator $\dd^{\prime}$ defined in \eqref{eq:differential_prime} is nilpotent, and hence $\ddh^{\prime}\ddh^{\prime} = 0$, $\ddv\ddv = 0$ and $\ddh^{\prime}\ddv + \ddv\ddh^{\prime} = 0$, so both $\ddh^{\prime}$ and $\ddv$ are differentials on $\bigwedge J\E$. Upon pullback to $\I$, for any $\alpha\in\bigwedge J\E$ one has
    \begin{equation}
        \ddh[\alpha] = [\ddh^{\prime}\alpha]\quad\text{and}\quad \ddv[\alpha] = [\ddv\alpha]\,.
    \end{equation}
    \item[3)] $\dd^{\prime}$ has a well-defined action on $\bigwedge_{T} J\E$. 
\end{itemize}
\end{lem}
\begin{proof}
     The assertion {\it (1)} is verified directly: one starts from $\LD_{T-Z} = \dd i_{T-Z} + i_{T-Z} \dd$ and recalls that $i_{T}\vartheta = -1$. For {\it (2)} recall that $\dd \vartheta = 0$ and $\LD_{Z}\vartheta = 0$, and note that $\ddh^{\prime}$ and $\ddv$ are the bidegree-$(1,0)$ and -$(0,1)$ components of $\dd^\prime$. Also note that $\LD_{Z}\alpha$ vanishes on $\I$ for any $\alpha\in \bigwedge J\E$, so $[\vartheta\wedge\LD_{Z}\alpha] = 0$. To prove {\it (3)}, note the following consequences of \eqref{eq:Cartan_prime} which imply that $\dd^{\prime}$ acts on $\bigwedge_{T} J\E$: $i_T \dd^{\prime} = -\dd^{\prime} i_T + \LD_{T-Z}$ and $\dd^{\prime} \LD_{T-Z} = \LD_{T-Z} \dd^{\prime}$.
\end{proof}
As a consequence of the points {\it (2)}, {\it (3)} of the above Lemma, $\bigwedge_{T} J\E$ is a bi-complex:
    \begin{equation}\label{eq:bicomplex_T}
    \begin{array}{ccccccccc}
        \bigwedge^{(0,0)}_T J\E & \xrightarrow{\ddh^{\prime}} & \bigwedge^{(1,0)}_T J\E & \xrightarrow{\ddh^{\prime}} & \dots & \xrightarrow{\ddh^{\prime}} & \bigwedge^{(n,0)}_{T} J\E& \xrightarrow{\ddh^{\prime}} & 0\\ 
        \downarrow \scriptstyle{\ddv} & \hfill & \downarrow \scriptstyle{\ddv} & \hfill & \hfill & \hfill & \downarrow \scriptstyle{\ddv}& \hfill & \hfill \\
        \bigwedge^{(0,1)}_T J\E & \xrightarrow{\ddh^{\prime}} & \bigwedge^{(1,1)}_T J\E & \xrightarrow{\ddh^{\prime}} & \dots & \xrightarrow{\ddh^{\prime}} & \bigwedge^{(n,1)}_T J\E& \xrightarrow{\ddh^{\prime}} & 0\\
        \downarrow \scriptstyle{\ddv} & \hfill & \downarrow \scriptstyle{\ddv} & \hfill & \hfill & \hfill & \downarrow \scriptstyle{\ddv}& \hfill & \hfill\\
        \dots & \hfill & \dots & \hfill & \hfill & \hfill & \dots & \hfill & \hfill\\
    \end{array}
    \end{equation}
Note that the bi-complex of ambient T-forms is a deformation of the subalgebra of ambient forms which are in the common kernel of $i_{T}$ and $\LD_{T}$. Indeed, for any $\alpha \in \bigwedge J\E$ such that $i_{T}\alpha = 0$ and $\LD_{T}\alpha = 0$, one constructs $e^{t\LD_{Z}}\alpha \in\bigwedge_{T}J\E$. In particular, $\dd^{\prime} = e^{t\LD_{Z}}\dd e^{-t\LD_{Z}}$.
\begin{lem}\label{lem:representatives_for_equivariant_forms}
    One has $[\alpha]\in \bigwedge_{T}\I$ iff there is a representative $\tilde{\alpha}\in \bigwedge_{T}J\E$ in $[\alpha]$. Moreover, for $\ddh[\alpha]\in \bigwedge_{T}\I$ there exists $\tilde{\alpha}\in \bigwedge_{T}J\E$ such that
    \begin{equation}\label{eq:equivariant_exact}
        \ddh[\alpha] = [\ddh^{\prime}\tilde{\alpha}]\,.
    \end{equation}
\end{lem}
\begin{proof}
     It is easy to check that $i_{T}$ and $\LD_{T}$ are well defined on equivalence classes. Also note that $[\LD_{Z}\alpha] = 0$ for any $\alpha \in \bigwedge J\E$. So if $\alpha\in \bigwedge_{T} J\E$, then one has $\LD_{T} [\alpha] = [\LD_{T-Z}\alpha] = 0$ and $i_{T}[\alpha] = [i_{T}\alpha] = 0$. Other way around, given a form $[\alpha]\in \bigwedge_{T}\I$, one can lift it to $\bigwedge_{T}J\E$ by solving the conditions in \eqref{eq:difforms_subspace_strict}.
     \vskip 4 pt
     To prove the rest, for $\ddh[\alpha] \in\bigwedge_{T}\I$ one can find $[\alpha^{\prime}]\in \bigwedge_{T}\I$ such that $\ddh[\alpha] = \ddh[\alpha^{\prime}]$ (which is possible by Lemma \bref{lem:isomorphism_epsilon}). Then it is already proven that there exists $\tilde{\alpha}\in\bigwedge_{T} J\E$ such that $[\alpha^{\prime}] = [\tilde{\alpha}]$, so \eqref{eq:equivariant_exact} follows by assertion {\it (2)} of Lemma \bref{lem:d_prime}.
\end{proof}

\paragraph{Lagrangian $T$-forms.} By analogy with the usual bicomplex of differential forms on a jet bundle, we refer to the elements of the upmost right component $\bigwedge^{(n,0)}_{T}J\E$ of \eqref{eq:bicomplex_T} as {\it ambient Lagrangians}. Any ambient Lagrangian $\lambda$ can be written in terms of a Lagrangian density $L$ as:
\begin{equation}
    \lambda = L\,\vol_T\,,\;\;\qquad 
    L\in \tensor*{\bigwedge\nolimits}{^{(0,0)}} J\E\,,
\end{equation}
where $\vol_{T} = i_{T}\vol$ is lifted to $\bigwedge J\E$. The second condition in \eqref{eq:difforms_subspace_strict}, $\LD_{T-Z}\lambda = 0$, leads to the following constraint:
\begin{equation}\label{eq:Lagrangian_form}
    \LD_{T-Z}L + \mathrm{div}\,T\,L = 0\,,
\end{equation}
where divergence is defined for any total vector field $V = V^{A} D_{A}$ as follows: 
\begin{equation}\label{eq:divergence}
    \ddh i_{V}\vol = \mathrm{div}\, V\, \vol\quad\Rightarrow\quad \mathrm{div}\,V = \frac{1}{\rho}\,D_{A}(\rho\,V^{A})\,,
\end{equation}
where $\rho$ is the volume density in \eqref{volform}.
\vskip 4 pt

We call $\ddh^{\prime}$-exact ambient Lagrangians {\it trivial}. In order to describe trivial ambient Lagrangians, {\it i.e.} ambient Lagrangians of the form  $\lambda = \ddh^{\prime}\beta$, with $\beta \in \bigwedge^{(n-1,0)}_{T}J\E$, note that any ambient  $(n-1,0)$-$T$-form can be written as
\begin{equation}
    \beta = i_{H} \vol_{T} = H^{A}\,i_{D_{A}}\vol_{T}\,,
\end{equation}
where $H = H^{A} D_{A}$ is a total vector field on $J\E$. As we demonstrate in Appendix \bref{lem:Lagrangian_T_exact}, this implies:
\begin{equation}\label{eq:Lagrangian_h-exact}
    \ddh^{\prime}\beta = \mathrm{div}\, H_{\perp}\,\vol_T\,,\quad \text{where}\quad H_{\perp} = H + i_{H}\vartheta\,T\,.
\end{equation}
As a result, the Lagrangian density $L$ of a trivial ambient Lagrangian $\lambda$ has the form of an ambient divergence. Note that $\beta$ is unaffected if one varies $H$ by a total vector field proportional to $T$, so $\beta = i_{H_{\perp}}\vol_{T}$.

\subsection{Ambient description of the variational bicomplex of $JE$}

The jet-bundle $J\E$ is equipped with the Euler-Lagrange derivative
\begin{equation}\label{eq:EL_ambient}
    \delta : \bigwedge{}^{(n+1,0)} J\E \to \bigwedge{}^{(n+1,1)} J\E\,,
\end{equation}
such that for any $\alpha \in\bigwedge^{(n+1,0)} J\E$ one has $\ddv\alpha = \delta\alpha + \ddh \sigma$ for  some $(n,1)$-form $\sigma$, and $\delta(\ddh \beta) = 0$ for all $\beta \in \bigwedge^{(n,1)}J\E$. 
In the same way, the jet-bundle $JE$ is also equipped with the Euler-Lagrange derivative, for which we keep the same notation $\delta$.
\vskip 4 pt

Thanks to the isomorphism of Lemma \bref{lem:isomorphism_epsilon}, the map
$\delta : \bigwedge{}^{(n+1,0)} JE \to \bigwedge{}^{(n+1,1)} JE$ defines a map $ \bigwedge{}_T^{(n+1,0)} \I \to \bigwedge{}_T^{(n+1,1)} \I$. A remarkable fact is that this map can be expressed explicitly in terms of the canonical Euler-Lagrange derivative \eqref{eq:EL_ambient} in the ambient space, giving a lift of the variational calculus on $JE$ to the ambient space. First, note that for any $\alpha \in \bigwedge^{(n,j)}_{T} J\E$ there is a unique preimage with respect to the map $i_{T}$ which we denote $i_{T}^{-1}\alpha = \bigwedge^{(n+1,j)} J\E$. By representing $\alpha = H\wedge \vol_{T}$ (with $H\in \bigwedge^{(0,j)} J\E$), one has $i_{T}^{-1}\alpha = (-)^{j} H\wedge \vol$.
\vskip 4 pt

Define the {\it ambient Euler-Lagrange derivative}
\begin{equation}\label{eq:EL_derivative}
    \widehat{\delta}\,:\, \tensor*{\bigwedge\nolimits}{^{(n,0)}_{T}}J\E \to  \tensor*{\bigwedge\nolimits}{^{(n,1)}} J\E\,,\quad \widehat{\delta}\;:\;\lambda \;\mapsto\; i_{T} \delta (i^{-1}_{T}\lambda)\,.
\end{equation}
By representing $\lambda = L\,\vol_T$, one has the following explicit formula in local coordinates:
\begin{equation}\label{eq:EL_derivative_explicit}
    \widehat{\delta} \left(L\,\vol_T\right) = \ddv u^{\sfa}\wedge \sum_{q \geqslant 0} (-)^q\,\frac{1}{\rho}D_{A(q)}\left(\rho\,\partial_{\sfa}^{A(q)} L\right)\,\vol_T\,,
\end{equation}
where we have introduced the shorthand notation $D_{A(q)} = D_{A_1}\dots D_{A_q}$. The above expression coincides with the ordinary formula for the Euler-Lagrange derivative except for the volume-form part where instead of the ambient top-form $\vol$ we put $\vol_T$.
\vskip 4 pt

 For any $[\lambda]\in \bigwedge_{T}\I$, let  $\lambda\in \bigwedge_{T}J\E$ be its ambient $T$-form representative (which exists thanks to Lemma \bref{lem:representatives_for_equivariant_forms}).
 With the map \eqref{eq:EL_derivative} at hand, 
 we can now define 
\begin{equation}\label{eq:EL_on_I}
    \widehat{\delta}_{\I}\,:\, \tensor*{\bigwedge\nolimits}{^{(n,0)}_{T}}\I \to  \tensor*{\bigwedge\nolimits}{^{(n,1)}} \I\,,\quad \widehat{\delta}_{\I}\;:\;[\lambda]\; \mapsto\; [\widehat{\delta}\lambda]\,.
\end{equation}
We have the following Lemma, whose proof is relegated to 
Appendix \bref{sec:jet_proofs}:
\begin{lem}\label{lem:EL_maps}
    For any $\lambda \in \bigwedge^{(n,0)}_{T}J\E$, the following assertions hold:
    \begin{itemize}
        \item[{\it 1)}] $\widehat{\delta}\lambda$ belongs to $\bigwedge^{(n,1)}_{T} J\E$, hence
            \begin{equation}\label{eq:EL_T-in-form}
                \widehat{\delta} : \tensor*{\bigwedge\nolimits}{^{(n,0)}_{T}} J\E \rightarrow \tensor*{\bigwedge\nolimits}{^{(n,1)}_{T}} J\E\,.
            \end{equation}
        Moreover, $\widehat{\delta}(\ddh^{\prime} \beta) = 0$ for any $\beta \in \bigwedge^{(n-1,0)}_{T} J\E$.        \item[{\it 2)}] If $\lambda$ vanishes on $\I$, so does $\widehat{\delta}\lambda$. Therefore, \eqref{eq:EL_on_I} gives a well-defined map
            \begin{equation}\label{eq:EL_T_I-in-form}
                \widehat{\delta}_{\I} : \tensor*{\bigwedge\nolimits}{^{(n,0)}_{T}} \I \rightarrow \tensor*{\bigwedge\nolimits}{^{(n,1)}_{T}} \I\,.
            \end{equation}
    \end{itemize}
\end{lem}

Finally, we arrive at the main theorem, which relates the map \eqref{eq:EL_derivative} to the Euler-Lagrange derivative on $JE$.
\begin{thm}\label{thm:EL_derivative}
    One has the following commutative diagram
    \begin{equation}\label{eq:EL_diagram}
    \def\arraystretch{1.5}
    \begin{array}{ccccc}
            \tensor*{\bigwedge\nolimits}{^{(n,0)}_{T}} J\E & \xrightarrow{\mathfrak{i}^{*}_{\Delta}} & \tensor*{\bigwedge\nolimits}{^{(n,0)}_{T}} \I & \xrightarrow{\varepsilon_{\Sigma}} & \bigwedge^{(n,0)} J E \\
            \displaystyle\downarrow\scriptstyle{\widehat{\delta}} & \hfill & \displaystyle\downarrow\scriptstyle{\widehat{\delta}_{\I}} & \hfill & \displaystyle\downarrow\scriptstyle{\delta} \\
            \tensor*{\bigwedge\nolimits}{^{(n,1)}_{T}} J\E & \xrightarrow{\mathfrak{i}^{*}_{\Delta}} & \tensor*{\bigwedge\nolimits}{^{(n,1)}_{T}} \I & \xrightarrow{\varepsilon_{\Sigma}} & \bigwedge^{(n,1)} J E 
    \end{array}
    \end{equation}
    where $\varepsilon_{\Sigma}$ is the isomorphism described in Lemma \bref{lem:isomorphism_epsilon}.
\end{thm}

We will say that the bicomplex $\bigwedge_{T} J\E$, supplemented with the ambient Euler-Lagrange derivative \eqref{eq:EL_T-in-form}, is an {\it ambient variational bicomplex}. Constructing the extended ambient variational bicomplex, explicitly involving forms of higher vertical degrees, is an interesting problem, which lies however beyond the scope of the present work. 
\vskip 4 pt

Given an ambient Lagrangian $\lambda\in\bigwedge^{(n,0)}_{T}J\E$ one defines the corresponding ambient Euler-Lagrange equations of motion $e_{\sfa}(X,\boldsymbol{u})$ via the formula $\widehat{\delta} \lambda = \ddv u^{\sfa}\wedge e_{\sfa}(X,\boldsymbol{u})\,\vol_{T}$. A solution of the equations of motion is a section $\Phi\in \Gamma(\E)$ whose jet prolongation $j^{\infty}\Phi$ is a section of $\I \subset  J\E$ (hence $\left.(j^{\infty}\Phi)\right|_{\Sigma} = j^{\infty}\phi$ is an infinite jet prolongation of a section $\phi  = \left.\Phi\right|_{\Sigma}\in\Gamma(E)$ by Lemma \bref{lem:pullback_bundle_I}) which is such that $(j^{\infty}\Phi)^{*}e_{\sfa}(X,\boldsymbol{u}) = 0$. The clear advantages of this ambient-space variational calculus are its coordinate independence and the explicit expression for the ambient Euler-Lagrange derivative \eqref{eq:EL_derivative}. Thanks to the isomorphism $\varepsilon_{\Sigma}$, one can work in the ambient variational bicomplex postponing to solve the constraints and restriction to $\Sigma$ to the very end, exactly as  was done in Sections \bref{sec:massless} and \bref{sec:massive} for the ambient-space construction of the Lagrangians for totally-symmetric massive and (partially) massless higher-spin $AdS_{d+1}$ fields.

\begin{proof}[Proof of Theorem \bref{thm:EL_derivative}]
The left part of the diagram \eqref{eq:EL_diagram} reflects the definition \eqref{eq:EL_on_I}, and is commutative due to the assertion {\it (2)} of Lemma \bref{lem:EL_maps}. Since $\varepsilon_{\Sigma}$ is an isomorphism, the right part of the diagram \eqref{eq:EL_diagram} defines the operator
\begin{equation}\label{eq:EL_draft}
    \delta^{\prime} = \varepsilon_{\Sigma}\circ \widehat{\delta}_{\I}\circ \varepsilon_{\Sigma}^{-1}\; :\; \tensor*{\bigwedge\nolimits}{^{(n,0)}} JE \rightarrow \tensor*{\bigwedge\nolimits}{^{(n,1)}} JE\,.
\end{equation}
In order to prove that $\delta^{\prime} = \delta$, we will show that for any $\mathbb{L} \in \bigwedge^{(n,0)} JE$ there is $\mathcal{J}\in \bigwedge^{(n-1,1)} JE$ such that $\ddv \mathbb{L} = \delta^{\prime} \mathbb{L} + \ddh \mathcal{J}$, and $\delta^{\prime} \mathbb{L} = 0$ when $\mathbb{L}$ is $\ddh$-exact. Then $\delta^{\prime} = \delta$ follows from the  uniqueness of the Euler-Lagrange derivative on $\bigwedge JE$ (see, {\it e.g.}, \cite[Corollary 5.2]{Anderson_bicomplex}).
\vskip 4 pt

Any given Lagrangian $\mathbb{L}\in \bigwedge^{(n,0)}JE$ can be lifted to an equivalence class $[\lambda]\in \bigwedge_{T}\I$, such that $\varepsilon_{\Sigma}[\lambda] = \mathbb{L}$. Furthermore, by Lemma \bref{lem:representatives_for_equivariant_forms}, one can consider a representative $\lambda\in \bigwedge_{T}J\E$. Then from $\lambda = i_{T} i^{-1}_{T} \lambda$ and recalls \eqref{eq:properties_T_Z} to write $\ddv \lambda = i_{T} \ddv (i^{-1}_{T} \lambda)$. Since $i^{-1}_{T} \lambda \in \bigwedge^{(n+1,0)} J\E$, one has 
\begin{equation}\label{eq:variation_0}
    \ddv i^{-1}_{T} \lambda = \delta i^{-1}_{T} \lambda + \ddh \beta\,\quad\text{with some}\quad \beta \in \tensor*{\bigwedge\nolimits}{^{(n,1)}_{T}} J\E\,.
\end{equation}

Next, apply $\LD_{T-Z}$ to both sides of \eqref{eq:variation_0}, and note that $i_{T}^{-1}$ commutes with $\LD_{T-Z}$, which can be verified directly for any $\alpha = H\wedge\vol \in \bigwedge^{(n+1,0)} J\E$ and its $i_{T}$-image $i_{T}\alpha = H \wedge \vol_{T}$ (where $H\in \bigwedge^{(0,0)}J\E$). Also recall that $\delta$ commutes with prolongations of vector fields which are projectable on $\E$ (see, {\it e.g.}, \cite[Corollary 2.13]{Anderson_bicomplex}). All in all, one has
\begin{equation}\label{eq:L_T-Z-exact}
    \ddh(\LD_{T-Z}\beta) = 0\quad\Rightarrow\quad \LD_{T-Z}\beta = \ddh\sigma\quad\text{for some}\quad \sigma \in \tensor*{\bigwedge\nolimits}{^{(n-1,1)}} J\E\,,
\end{equation}
where we have used that horizontal cohomology is empty in the bi-degree $(n,1)$ (see, {\it e.g.}, \cite[Proposition 4.2]{Anderson_bicomplex}). By the assertions {\it (1)}, {\it (2)} of Lemma \bref{lem:d_prime} one obtains
\begin{equation}\label{eq:variation_exact_term}
    i_{T}\ddh^{\prime}\beta = \ddh \sigma - \ddh^{\prime}i_{T}\sigma\quad\Rightarrow\quad i_{T}\ddh[\beta] = \ddh [\sigma - i_{T}\beta]\,.
\end{equation}

Now, apply $i_{T}$ to both sides of \eqref{eq:variation_0},
\begin{equation}\label{eq:variation_1}
    \ddv \lambda = \widehat{\delta}\lambda + i_{T}\ddh \beta\,,
\end{equation}
then recall \eqref{eq:L_T-Z-exact} and perform the pullback of the so-obtained expression to $\I$: 
\begin{equation}\label{eq:variation_2}
   \ddv [\lambda] = \widehat{\delta}_{\I}[\lambda] + \ddh [\sigma - i_{T}\beta]\,,
\end{equation}
where we have used the definitions \eqref{eq:EL_derivative}, \eqref{eq:EL_on_I}. Note that $\ddh [\sigma - i_{T}\beta] \in \bigwedge^{(n,1)}_{T}\I$. Indeed, the first condition in \eqref{eq:forms_I_equiv} follows from \eqref{eq:variation_exact_term}, while for the second condition one evaluates $\LD_{T}$ on both sides of the latter expression in \eqref{eq:variation_2} and recalls that $\lambda\in\bigwedge^{(n,0)}_{T} J\E$:
\begin{equation}
    \LD_{T}\ddv[\lambda] = \ddv \LD_{T}[\lambda] = 0\,,\quad\text{and}\quad \LD_{T}\widehat{\delta}_{\I}[\lambda] = [\LD_{T-Z}\widehat{\delta}\lambda] = \widehat{\delta}_{\I}[\LD_{T-Z}\lambda] = 0\,.
\end{equation}
By Lemma \bref{lem:representatives_for_equivariant_forms} there exists $\omega \in \bigwedge^{(n,1)}_{T}J\E$ such that $\ddh [\sigma - i_{T}\beta] = \ddh[\omega]$. Finally, by evaluating $\varepsilon_{\Sigma}$ on \eqref{eq:variation_2} one obtains
\begin{equation}
    \ddv \mathbb{L} = \delta^{\prime}\mathbb{L} + \ddh\mathcal{J}\,,
\end{equation}
where $\mathcal{J} = \varepsilon_{\Sigma}[\omega]$, and one recalls that $[\lambda] = \varepsilon^{-1}_{\Sigma}\mathbb{L}$ and \eqref{eq:EL_draft}. To check that $\delta^{\prime}\mathbb{L} = 0$ when $\mathbb{L}$ is $\ddh$-exact, note that by Lemma \eqref{lem:representatives_for_equivariant_forms} one has $\varepsilon^{-1}_{\Sigma}\mathbb{L} = [\ddh^{\prime}\beta]$ for some $\beta \in \bigwedge^{(n-1,0)}_{T} J\E$. Then $\delta^{\prime}\mathbb{L} = \varepsilon_{\Sigma}[\widehat{\delta}(\ddh^{\prime}\beta)] = 0$ by the assertion {\it (1)} of Lemma \bref{lem:EL_maps}.
\end{proof}

\subsection{Examples}\label{sec:examples}

Let us apply the proposed general formalism to the ambient description of totally-symmetric fields on $AdS_{d+1}$. In this case one takes the ambient space to be $\hM = \mathbb{R}^{2,d}_{+}$ endowed with the flat metric, and the fundamental vector field is the ambient Euler vector field $T = X\cdot\partial_X$. An embedding $\Sigma = AdS_{d+1}\hookrightarrow\hM =\mathbb{R}^{2,d}_{+}$ is set by fixing $r = \sqrt{-X^2} = \ell$. The volume form $\vol$ (and $\vol_{X}$) are defined as in \eqref{eq:vol_AdS}, and the Ehresmann connection one-form is $\vartheta = -\frac{X_A \,\dd X^{A}}{X^2}$. Note that the Cartesian coordinates $X^A$ on $\mathbb{R}^{2,d}_{+}$ verify \eqref{eq:homogeneous_coordinates} with $w = 1$, so the jet prolongation of the weight constraint \eqref{eq:constraint_homogeneity} reads as \eqref{eq:useful_formulae_hom_coord}.

\paragraph{Scalar field.} The simplest example is provided by the ambient scalar field subject to the homogeneity constraint \eqref{eq:constraint_homogeneity}. Due to Theorem \bref{thm:massive}, the corresponding ambient Lagrangian expressed in field-theoretical terms reads
\begin{equation}\label{eq:example_L_scalar}
    \Lx[\Phi] = \frac{1}{2}\, r^{-(d-2\Delta)}\,\Phi\Box\Phi\,\vol_{X}\,.
\end{equation}
In order to apply the above general scheme, take the ghost-degree-$0$ projection $\E = \F^{(0)}_{0} \subset \F^{(0)}$ (recall the comment below \eqref{eq:helicity_ghosts}). Fibers of the jet-bundle $J\E$ are parametrised by $\boldsymbol{u} = \{u_{A(q)}\;:\; q\geqslant 0\}$. The sub-bundle $\I$ is singled out by \eqref{eq:jet_sub-bundle}.
\vskip 4 pt

From the local $(d+1)$-form \eqref{eq:example_L_scalar} on $\mathbb{R}^{2,d}_{+}$ one constructs the corresponding local $(d+1,0)$-$T$-form on $J\E$:
\begin{equation}
    \lambda = \frac{1}{2}\,r^{-(d-2\Delta)} \,u\, u^A{}_A\,\vol_{X}\,,
\end{equation}
which indeed satisfies \eqref{eq:difforms_subspace_strict}. 
Application of the ambient Euler-Lagrange derivative \eqref{eq:EL_derivative} leads to
\begin{equation}
    \widehat{\delta}\lambda = \ddv u\wedge\frac{1}{2}\,\left(r^{-(d-2\Delta)}u^A{}_A + D^A D_A\left(r^{-(d-2\Delta)}u\right)\right)\,\vol_{X}\,.
\end{equation}
Taking into account that $D_Af(X) = \partial_A f(X)$ and using the relations \eqref{eq:normal_diff} together with \eqref{eq:total_vec_basis} one finally arrives at the ambient equations of motion
\begin{equation}
    \widehat{\delta}\lambda = \ddv u \wedge e(X,\boldsymbol{u})\,\vol_{X} = \ddv u\wedge r^{-(d-2\Delta)} \left(u^A{}_A + \frac{(d - 2\Delta)}{r^2}\,(X^{B} u_{B} + \Delta\, u)\right)\,\vol_{X}\,.
\end{equation}
For any section $\Phi\in\Gamma(\E)$ which satisfies \eqref{eq:constraint_homogeneity}, the equation $(j^{\infty}\Phi)^{*} e(X,\boldsymbol{u}) = 0$ is equivalent to $\Box \Phi = 0$, so one reproduces the ambient description of a scalar field on $AdS_{d+1}$.
\vskip 4 pt

Note another ambient Lagrangian $\lambda^{\prime}$ obtained from \eqref{eq:example_L_scalar} via integration by parts:
\begin{equation}
\def\arraystretch{1.4}
\begin{array}{rl}
    \lambda^{\prime} = & \displaystyle -\frac{1}{2}\,r^{-(d-2\Delta)}\,\left(u_A u^A - \frac{\Delta(d-2\Delta) }{r^2}\,u^2\right)\,\vol_{X} \\
    = & \lambda + \ddh^{\prime} \mathcal{J} + \zeta\,,
\end{array}
\end{equation}
where 
\begin{equation}
    \mathcal{J} = - J_{\perp}^{A} \,i_{D_{A}}\vol_{X}\quad\text{with}\quad J^{A}_{\perp} = \frac{1}{2} r^{-(d-2\Delta)}\left(u\,u^{A} -\frac{\Delta}{r^2}\, X^{A} u^2\right)\,,
\end{equation}
and
\begin{equation}
    \zeta = \frac{d-4\Delta}{2r^2}\,r^{-(d-2\Delta)} u(X^A u_{A} + \Delta \,u)\,\vol_{X} + \frac{1}{2} D_{A}\left(r^{-(d-2\Delta + 2)} X^{A}u(X^{B}u_{B} + \Delta u)\right)\wedge \vol_{X}\,.
\end{equation}
One can check that $\mathcal{J}$ is a $(d,0)$-$T$-form, therefore both $\ddh^{\prime} \mathcal{J}$ and $\zeta$ are $(d+1,0)$-$T$-forms as well. Moreover, $\zeta$ vanishes by virtue of \eqref{eq:constraint_homogeneity}. One can verify that evaluation of the ambient Euler-Lagrange derivative using $\lambda^{\prime}$ leads to the same ambient equations of motion.

\paragraph{Massless higher-spin fields.} Let us also consider the case of a massless spin-$s$ field ($s\geqslant 2$) in the triplet formulation. The same lines of reasoning apply to massive and partially massless higher-spin fields thanks to Theorem \bref{thm:massive}.
\vskip 4 pt

Along the same lines as for the scalar field, we fix the ghost-degree-$0$ projection $\E = \F^{(s)}_{0} \subset \F^{(s)}$ with its fibers parametrised by the collection $u = \{u_{A(s)}$, $u_{A(s-1)}, u_{A(s-2)}\}$ (standing for the triplet components $B$, $C$ and $D$ respectively, see \eqref{eq:ghost_expansion}). The above bundle is promoted to its infinite jet extension $J\E$ with fibers parametrised by the coordinates $\boldsymbol{u} = \{u_{A(s-k) ; B(q)}\;:\; k = 0,1,2\,,\;\; q \geqslant 0\}$ designating derivatives of the fields.
\vskip 4 pt

The jet prolongation of weight constraints \eqref{eq:homogeneity_ghosts} reads, for all $q \geqslant 0$, as
\begin{equation}\label{eq:constraints_weight_example_s}
\begin{array}{l}
    X^{C} u_{A(s);CB(q)} + (2 - s + q)\, u_{A(s);B(q)} = 0\,,\\
    X^{C} u_{A(s-1);CB(q)} + (3 - s + q)\, u_{A(s-1);B(q)} = 0\,,\\
    X^{C} u_{A(s-2);CB(q)} + (2 - s + q)\, u_{A(s-2);B(q)} = 0\,.
\end{array}
\end{equation}
The above constraints lead to the sub-bundle $\I \hookrightarrow J\E$. We omit detailed consideration of the algebraic constraints \eqref{eq:ghost_traceless} and \eqref{eq:ghost_tangent}, treating them as already imposed.
\vskip 4 pt

The Lagrangian $\mathbb{L}[\Phi]$ in Theorem \bref{thm:massless}, understood as a $(d+1,0)$-$T$-form, reads as
\begin{equation}\label{eq:Lagrangian_jet_s}
\def\arraystretch{1.4}
\begin{array}{rl}
    \lambda = & \displaystyle\frac{1}{2}\, r^{-(d + 2s - 4)} \bigg(\frac{1}{s!}\,u^{A(s)}u_{A(s);B}{}^{B} - \frac{1}{(s-1)!}\, u^{A(s-1)} u_{A(s-1)} -\frac{1}{(s-2)!}\, u^{A(s-2)} u_{A(s-2);B}{}^{B} \\
    \hfill & + \displaystyle\frac{1}{(s-1)!}\,\big(u^{A(s)} u_{A(s-1);A} - u^{A(s-1)} u_{A(s-1)B;}{}^{B}\big) \\
    \hfill & + \displaystyle\frac{1}{(s-2)!}\,\big(u^{A(s-2)} u_{A(s-2)B;}{}^{B} - u^{A(s-1)} u_{A(s-2);A}\big)\bigg)\,\vol_{X}\,,
\end{array}
\end{equation}
where the conditions \eqref{eq:difforms_subspace_strict} hold for each term.
\vskip 4 pt

The structure of the constraints fits the general construction, and thus Theorem \bref{thm:EL_derivative} applies. The ambient equations of motion read as
\begin{equation}\label{eq:EL_example_s}
\def\arraystretch{1.4}
\begin{array}{rl}
    \widehat{\delta}\lambda = & \displaystyle\sum_{k = 0,1,2} \ddv u^{A(s-k)}\wedge e_{A(s-k)}(X,\boldsymbol{u})\, \vol_{X} \\
    = & \displaystyle \ddv u^{A(s)}\wedge \frac{1}{s!}\, r^{-(d-2s+4)}\big(u_{A(s);B}{}^{B}  + s\, u_{A(s-1);A}\big)\,\vol_{X} \\
    - & \displaystyle \ddv u^{A(s-1)} \wedge \frac{1}{(s-1)!}\, r^{-(d-2s+4)}\big(u_{A(s-1)B;}{}^{B} + (s-1) u_{A(s-2);A} + u_{A(s-1)} \big)\,\vol_{X}\\
    - & \displaystyle \ddv u^{A(s-2)} \wedge \frac{1}{(s-2)!}\,r^{-(d + 2s - 4)} \big(u_{A(s-2);B}{}^{B} - u_{A(s-2)B;}{}^{B}\big)\,\vol_{X}\\
    + & \displaystyle \ddv u^{A(s)}\wedge \frac{1}{s!}\,r^{-(d-2s+4)} \big(X^{B}u_{A(s);B} + (2-s)\, u_{A(s)}\big)\,\vol_{X}\\
    - & \displaystyle \ddv u^{A(s-2)}\wedge \frac{1}{(s-2)!}\,r^{-(d-2s+4)} \big(X^{B}u_{A(s-2);B} + (2-s)\, u_{A(s-2)}\big)\,\vol_{X}\,.
\end{array}
\end{equation}
For any section $\Phi \in \Gamma(\E)$, whose components verify \eqref{eq:subspaces}, \eqref{eq:ghost_traceless} and \eqref{eq:ghost_tangent}, the equations $(j^{\infty}\Phi)^{*} e_{A(s-k)}(X,\boldsymbol{u}) = 0$ (with $k = 0,1,2$) are equivalent to the correct equations of motion $\Omega\Phi = 0$, provided by the BRST operator \eqref{eq:BRST_operator}. Note that the last two lines in \eqref{eq:EL_example_s} are proportional to the constraints \eqref{eq:constraints_weight_example_s}, and thus vanish on $\I$.

\paragraph{Massive fields via flat dimensional reduction.}

As another example, let us briefly note how the above formalism applies to the ambient Lagrangian formulation of massive fields of arbitrary symmetry in flat spacetime via dimensional reduction of massless fields in ambient space (see \cite{Alkalaev:2008gi,Chekmenev:2019ayr} and \cite{Metsaev:2012uy,Metsaev:2017cuz} for description of massless fields and dimensional reduction in BRST terms). For the ambient space one takes $\hM = \mathbb{R}^{2,d}$ with the flat metric, and chooses $T = \partial_{0^{\prime}}$. The embedding $\Sigma = \mathbb{R}^{1,d} \hookrightarrow \hM$ is singled out by $X^{0^{\prime}} = 0$, so $\vartheta = -\dd X^{0^{\prime}}$. We will write $\vol_{0^\prime} = i_{\partial_{0^{\prime}}} \vol$.
\vskip 4 pt

The bundle $\E\to \mathbb{R}^{2,d}$ is constructed for a particular field content in question: the components of the field multiplet $\{\Phi^{\sfa}\}$ are chosen according to \cite{Chekmenev:2019ayr} for any particular mixed-symmetry representation. The differential constraint \eqref{eq:weight_constraint_general} takes the form:
\begin{equation}\label{eq:mass_constraint_general}
    \left(\frac{\partial}{\partial X^{0^{\prime}}} + m\right)\,\Phi^{\sfa} = 0\,,\quad\text{for fixed}\quad m\in\mathbb{R}\,.
\end{equation}
The rest of the constraints are algebraic and commute with \eqref{eq:mass_constraint_general} (see \cite{Chekmenev:2019ayr}, for the particular case of totally-symmetric fields see the example of the flat limit presented in Section \bref{sec:flat_limit}).
\vskip 4 pt

The bundle $\E$ is promoted to its infinite jet extension $J\E$, while constraints are prolonged to their infinite jet extensions. As usual, we omit the analysis of algebraic constraints, and write only the jet prolongation of \eqref{eq:mass_constraint_general}:
\begin{equation}\label{eq:mass_constraint_general_jet}
    z^{\sfa}_{A(q)} := u^{\sfa}_{0^{\prime}\,A(q)} + m\,u^{\sfa}_{A(q)} = 0\quad\text{for all integer}\quad q \geqslant 0\,.
\end{equation} 
Jet prolongations of the constraints define the sub-bundle $\I \subset J\E$, whose pullback to $X^{0^{\prime}} = 0$ is isomorphic to the jet bundle $JE$ of the corresponding massive field in flat space $\mathbb{R}^{1,d}$.  
\vskip 4 pt

For the simplest example of a massive scalar field one has the following ambient Lagrangian
\begin{equation}
    \lambda = \frac{1}{2} \,e^{2m X^{0^\prime}} u \,u^{A}{}_{A}\,\vol_{0^\prime} \in \tensor*{\bigwedge\nolimits}{^{(d+1,0)}_{T}} J\E\,,
\end{equation}
which leads, according to Theorem \bref{thm:EL_derivative}, to the following equations of motion:
\begin{equation}
    \widehat{\delta}\lambda = \ddv u \wedge e(X,\boldsymbol{u})\,\vol_{0^\prime} = \ddv u\wedge e^{2m X^{0^\prime}} \big(u^{A}{}_{A} + 2m\,u\,(u_{0^\prime} + m\,u)\big)\,\vol_{0^\prime}\,.
\end{equation}
The second term in $e(X,\boldsymbol{u})$ is proportional to the constraints \eqref{eq:mass_constraint_general_jet}, so it vanishes on any section $\Phi\in\E$, which verifies \eqref{eq:mass_constraint_general}. As a result, the equation $(j^{\infty}\Phi)^{*}e(X,\boldsymbol{u}) = 0$ is equivalent to $\Box\Phi = 0$ in $\mathbb{R}^{2,d}$. By virtue of the constraint \eqref{eq:mass_constraint_general}, the latter equation reduces to the Klein-Gordon equation on $\mathbb{R}^{1,d}$ with the mass $m$.
\vskip 4 pt

For the case of totally-symmetric spin-$s$ fields (with $s\geqslant 1$), one reads off the corresponding Lagrangian $\lambda \in \bigwedge^{(d+1,0)}_{T} J\E$ from $\mathbb{L}^{\prime}[\Phi]$ in \eqref{eq:Lagrangian_flat} (with $\bar{X}^{0^{\prime}} \mapsto X^{0^{\prime}}$). In this case the field multiplet coincides with that of totally-symmetric massless AdS fields, where the constraint \eqref{eq:mass_constraint_general} is applied to all components. Finally, for an arbitrary mixed-symmetry massive field in $\mathbb{R}^{1,d}$ one simply takes the BRST description of a massless field of the same symmetry type in $\mathbb{R}^{2,d}$ \cite{Alkalaev:2008gi}, and reads off the corresponding ambient Lagrangian $\lambda\in\bigwedge^{(d+1,0)}_{T} J\E$ from the expression $\big(\Phi,\Omega\Phi\big)^{\prime}$ (with the inner product as in \eqref{eq:Lagrangian_flat}). The pullback of the latter Lagrangians to the sub-bundle $\left.\I\right|_{\Sigma}$, where mass constraints \eqref{eq:mass_constraint_general} are imposed and $X^{0^{\prime}} = 0$, reproduces the Lagrangians presented in \cite{Chekmenev:2019ayr} (see \cite{Metsaev:2012uy,Metsaev:2017cuz} for an alternative consideration).

\section*{Acknowledgements}

We thank the organisers of the APCTP workshop 
``Higher-Spin and applications" in Pohang and the organisers of the 
 MITP workshop ``Higher Structures, Gravity and Fields" in Mainz. 
The work of N.B. was partially supported by the F.R.S.-FNRS PDR grant 
``Fundamental issues in extended gravity'' No. T.0022.19.

\appendix
\newpage

\section{Adapted coordinates and covariant differentiation}\label{sec:coordinates}
\paragraph{Local AdS geometry.}
The embedding \eqref{eq:embedding} provides a coordinate system in the ambient space, which is adapted to embeddings $AdS_{d+1}\subset \mathbb{R}^{2,d}$. Namely, the parameter $r = \sqrt{-X^2}$ in \eqref{eq:embedding} can be supplemented by ``angular'' coordinates $x^{\mu}$ ($\mu = 1,\ldots, d+1$) parametrizing $AdS_{d+1}$ for any particular fixed value of $r$. Coordinates $x^\mu$ are homogeneous of degree $0$, {\it i.e.} $(X\cdot \partial_X) x^\mu = (r\partial_r)\,x^\mu = 0$. Parametrization of ambient space by $AdS_{d+1}$ radius $r$ and local $AdS_{d+1}$ coordinates $x^\mu$ gives rise to smooth functions $X^{A}(r,x^\mu)$ defining the vielbein field
\begin{equation}\label{eq:vielbein}
    e^{A}_\mu = \frac{\partial X^A}{\partial x^{\mu}}
\end{equation}
such that the $AdS_{d+1}$ induced metric is
\begin{equation}\label{eq:metric}
    g_{\mu\nu} = \eta_{AB}\,e^{A}_\mu e^{B}_\nu.
\end{equation}
In order to have a basis in $T\mathbb{R}^{2,d}$, we supplement the $d+1$ ambient vectors $e^A_\mu$ by another vector $n^A$ defined as 
\begin{equation}\label{eq:normal}
    n^A = \frac{X^A}{r}
\end{equation}
and thus normalized as $n_A n^A = -1$. As soon as $r$ is constant for any embedding of $AdS_{d+1}$ one has $X_A \partial_\mu X^A = 0$, and thus 
\begin{equation}
    n_A\,e^A_\mu = 0
\end{equation}
(up and down positions of ambient indices are swapped by the ambient metric $\eta_{AB}$). Vector $n^A$ admits the following two representations:
\begin{equation}\label{eq:normal_diff}
    n_A = -\frac{\partial\, r}{\partial X^A} = \frac{\partial X^A}{\partial r}.
\end{equation}
The first expression is obtained by differentiating \eqref{eq:embedding}, while the second one reflects the fact that $X^A$ has homogeneity degree $+1$, {\it i.e.} $(r\partial_r)\,X^A = X^A$.
\vskip 4 pt

The pushforward of the inverse of the $AdS_{d+1}$ metric \eqref{eq:metric} to the ambient space defines
\begin{equation}\label{eq:projector}
    \PP^{AB} = g^{\mu\nu}\,e^A_\mu e^B_\nu,
\end{equation}
which is a rank-$(d+1)$ projector with the kernel spanned by $n^A$:
\begin{equation}
    \PP^{A}{}_{B}\,\PP^{B}{}_{C} = \PP^{A}{}_{C}\,,\quad \PP^{A}{}_{B}e^B_{\mu} = e^{A}_{\mu}\,,\quad \PP^{AB}n_B = 0\,.
\end{equation}
As a result, the inverse of the ambient metric is decomposed into its tangent and normal parts with respect to $AdS_{d+1}$ as follows (also giving a decomposition of unity):
\begin{equation}\label{eq:decomposition_unity}
    \eta^{AB} = \PP^{AB} - n^A n^B\quad \Leftrightarrow\quad \delta^A{}_B = \PP^A{}_B - n^A n_B.
\end{equation}

Let us also introduce the dual vielbein
\begin{equation}\label{eq:dual_vielbein}
    e_A^\mu : = g^{\mu\nu}\eta_{AB}\,e^B_\nu\,,
\end{equation}
which solves the system of $(d+1)^2 + (d+2)^2$ equations
\begin{equation}\label{eq:dual_system}
    \left\{
        \begin{array}{l}
            \tilde{e}^\mu_A\, e^A_\nu = \delta^\mu_\nu\\
            e^A_\nu\, \tilde{e}^\nu_B = \PP^A{}_{B}
        \end{array}
    \right.
\end{equation}
imposed on $(d+1)(d+2)$ variables $\tilde{e}^\mu_A$. Other way around, the system \eqref{eq:dual_system} is not independent: it is subject to the trace constraint $\tilde{e}^{\nu}_B\cdot e^B_{\nu} = d+1$, as well as $(d+1)(d+2)$ constraints $e^A_\mu(\tilde{e}^\mu_B\cdot e^B_\nu-\delta^\mu_\nu) = (e^A_\rho\cdot \tilde{e}^\rho_B  - \PP^A{}_{B})e^B_\nu$. As a result, one has $(d+1)^2 + (d+2)^2 - (d+1)(d+2) - 1 = (d+1)(d+2)$ independent equations, thus leading to the uniqueness of the solution \eqref{eq:dual_vielbein}. 
\vskip 4 pt

Finally, let us show, that 
\begin{equation}\label{eq:dual_diff}
    e^\mu_A = \frac{\partial x^\mu}{\partial X^A}\,.
\end{equation}
To do that, we substitute 
\begin{equation}\label{eq:dX}
\dd X^A = e^A_\mu\,\dd x^\mu + n^A\dd r
\end{equation}
into 
\begin{equation}\label{eq:one}
        \dd x^\mu = \frac{\partial x^\mu}{\partial X^A}\,\dd X^A, \quad \text{which gives}\quad 
        \dd x^\mu = \frac{\partial x^\mu}{\partial X^A}\,e^A_\nu\,\dd x^\nu 
\end{equation}
because $n^A\,\partial_A x^\mu = \partial_r\,x^\mu = 0$. On the other hand, substituting $\dd r = \frac{\partial r}{\partial X^A}\,\dd X^A = -n_A\,\dd X^A$ (where we have made use of \eqref{eq:normal_diff}) into \eqref{eq:dX} one gets
\begin{equation}\label{eq:two}
    e^A_\mu\,\frac{\partial x^\mu}{\partial X^B} = \delta^A{}_B + n^A n_B = \PP^A{}_{B}\quad\text{agreeing to \eqref{eq:decomposition_unity}}.
\end{equation}
Expressions \eqref{eq:one} and \eqref{eq:two} together imply, that $\frac{\partial x^\mu}{\partial X^A}$ solves \eqref{eq:dual_system}. Thus, due to the uniqueness of the solution, the equality \eqref{eq:dual_diff} takes place.

\paragraph{Covariant differentiation on $AdS_{d+1}$.}\label{curved_geometry}

Differentiation in the ambient space induces covariant differentiation on $AdS_{d+1}\subset \mathbb{R}^{2,d}$ via the straightforward relation ({\it cf.} \eqref{eq:vielbein} and \eqref{eq:normal_diff})
\begin{equation}\label{eq:ambient_derivative}
    \partial_A = e^\mu_A\,\partial_\mu - \frac{n_A}{r}\,(r\,\partial_r)\,.
\end{equation}
The following expressions will be useful in the sequel:
\begin{equation}\label{eq:covar_1}
    \partial_\mu n^A = \partial_r e^A_\mu = \frac{1}{r}\,e^A_\mu.
\end{equation}
The first equality in \eqref{eq:covar_1} is due to that $n^A = \frac{\partial X^A}{\partial r}$ and $\partial_\mu\partial_r = \partial_r\partial_\mu$. The second one is obtained from applying $\partial_\mu$ to $n^A = \frac{X^A}{r}$.
\vskip 4 pt

The following relation
\begin{equation}\label{eq:covar_2}
    \partial_r n^A = 0
\end{equation}
is due to the $0$-degree homogeneity of $n^A$ \eqref{eq:normal}. 
\vskip 4 pt

Expression $\partial_r e^\mu_A$ for the dual of the vielbein field is obtained through the decomposition in terms of the local basis $\left\{n^A,e^A_\mu\right\}$ in the ambient space:
\begin{equation}
\begin{split}
    n^C\partial_r e^\mu_C = \partial_r(n^C\cdot e^\mu_C) = 0\,,\\
    e^C_\nu\,\partial_r e^\mu_C = -\partial_r e^C_\nu\cdot e^\mu_C = -\frac{1}{r}\,e^C_\nu e^\mu_C = -\frac{1}{r}\,\delta^\mu_\nu\,,
\end{split}
\end{equation}
which leads to
\begin{equation}\label{eq:covar_3}
    \partial_r e^\mu_A = -\frac{1}{r}\,e^\mu_A\,.
\end{equation}

Projector \eqref{eq:projector} allows us to define the tangent derivative \cite{Bekaert_Meunier}
\begin{equation}\label{eq:derivative}
    \DD_A = \PP\circ \partial_A \circ\PP\,,
\end{equation}
implying that tensors are projected to tangent ones before taking the ambient derivative, as well as afterwards. Tangent derivative \eqref{eq:derivative} obeys Leibniz rule and transforms type-$(m,n)$ ambient tensors to the type-$(m,n+1)$ ones, thus being a covariant derivative in $AdS_{d+1}$. It is compatible with the induced metric $g_{\mu\nu}$, because one can check that $\DD_C \PP_{AB} = 0$. Therefore for any tangent to $AdS_{d+1}$ tensor one has
\begin{equation}
    \DD_C \,T^{A_1\ldots}{}_{B_1\ldots} = \DD_C\big(e^{A_1}_{\mu_1}\ldots\cdot e^{\nu_1}_{B_1}\ldots\,T^{\mu_1\ldots}{}_{\nu_1\ldots}\big) = e_C^\lambda\cdot e^{A_1}_{\mu_1}\ldots\cdot e^{\nu_1}_{B_1}\ldots\,\nabla_\lambda T^{\mu_1\ldots}{}_{\nu_1\ldots},
\end{equation}
where $\nabla_\mu$ stands for the covariant derivative in $AdS_{d+1}$ corresponding to the Levi-Civita connection. One can check by acting on a vector field that Christoffel symbols are given by
\begin{equation}
    \Gamma^\lambda_{\mu\nu} = e^\lambda_C\,\partial^{}_\mu e^C_\nu \, ,
\end{equation}
involving the tangent component of the derivative $\partial_\mu e^C_\nu$. Normal component is obtained with the aid of \eqref{eq:covar_1}:
\begin{equation}
   \partial_{\mu}(n_C\,e^{C}_{\nu}) = 0\quad\Rightarrow\quad n_C \partial^{}_\mu e^C_\nu = -\frac{1}{r}\,g_{\mu\nu}\,,
\end{equation}
which gives
\begin{equation}
    \partial^{}_\mu e^C_\nu = e^C_\lambda\,\Gamma^\lambda_{\mu\nu} + \frac{1}{r}\, n^C\,g_{\mu\nu}\quad\Leftrightarrow\quad \nabla^{}_\mu e_\nu^C = \frac{1}{r}\,n^C\,g_{\mu\nu}\,.
\end{equation}
For the dual vielbein one gets, along the same lines as for \eqref{eq:covar_3} by applying \eqref{eq:covar_1} and $\nabla_\mu n^C = \partial_\mu n^C$,
\begin{equation}\label{eq:covar_4}
    \nabla^{}_\lambda e^\mu_C = \frac{1}{r}\,n_C\,\delta^\mu_\lambda \,.
\end{equation}
The following two relations can be useful for computations:
\begin{equation}\label{eq:relations}
    \partial_A n_B = \frac{1}{r}\,\PP_{AB},\quad \partial^{}_A e^\mu_B = - e^\lambda_A e^\nu_B\,\Gamma^\mu_{\lambda\nu} + \frac{1}{r}\,(n^{}_A e^\mu_B + n^{}_B e^\mu_A)\,.
\end{equation}

Finally, note that any ambient tensor field containing no ``naked'' AdS indices is differentiated as an AdS scalar, therefore for the tensors in question one rewrites \eqref{eq:ambient_derivative} in the form
\begin{equation}\label{eq:ambient_covariant_derivative}
    \partial_A = e^\mu_A\,\nabla_\mu - \frac{n_A}{r}\,(r\,\partial_r)\,.
\end{equation}
In particular, for the tangent and radial projections of the auxiliary variables $P^A$
\begin{equation}\label{eq:AdS_ocsillators}
    p^{\mu} := e_A^{\mu}P^A\quad \text{and}\quad q := -n_A P^A
\end{equation}
one can derive, by using \eqref{eq:covar_1}, \eqref{eq:covar_4} and $\nabla_{\mu} P^A = 0$, that
\begin{equation}\label{eq:cov_D_pq}
    \nabla_{\mu}p^{\nu} = -\frac{1}{r}\,q\,\delta^{\nu}_{\mu}\,,\quad \nabla_{\mu}q = -\frac{1}{r}\,p_{\mu}\,.
\end{equation}

\section{Ghost decomposition and BRST operator\\ in the adapted coordinates}\label{sec:formulae}

In order to perform calculations in terms of AdS fields we will make use of the adapted coordinates $r,x^{\mu}$ in the ambient space. In this respect, components $\Phi^{(m+n|n)}\in \mathcal{S}^{(m+n)}_{\Delta}$ of the radial decomposition \eqref{eq:radial_decomposition} are expressed as follows:
\begin{equation}\label{eq:radial_decomposition_AdS}
    \Phi^{(m+n)}_{n}\left(X\middle|P\right) = \phi^{(m,n)}_{\mu(m)}(X)\,\frac{p^{\mu(m)}}{m!}\frac{q^{n}}{n!}\quad\text{with}\quad \big(X\cdot\partial+ (\Delta - m)\big)\phi^{(m,n)}_{\mu(m)}(X) = 0\,.
\end{equation}
Despite the latter constraint can be resolved by an AdS tensor field $\phi^{(m,n)}_{\mu(m)}(x)$ multiplied by $r^{-(\Delta-m)}$ we prefer to keep  weight factor absorbed in $\phi^{(m,n)}_{\mu(m)}(X)$. As a shorthand notation for the AdS fields, instead of $\Phi^{(m+n)}_{n}\left(X\middle|P\right)$ we will write (omitting $X$- and $P$-dependence)
\begin{equation}
    \phi^{(m,n)}:= \phi^{(m,n)}_{\mu(m)}\,\frac{p^{\mu(m)}}{m!}\frac{q^{n}}{n!}\,,
\end{equation}
and also 
\begin{equation}\label{eq:field_components_Fock}
    \phi^{(m^{\prime},n^{\prime})}\psi^{(m,n)} := \langle\phi^{(m^{\prime},n^{\prime})},\psi^{(m,n)} \rangle_{\text{Fock}} = (-)^{n}\,m!\,n!\;\delta_{m^{\prime},m}\delta_{n^{\prime},n}\,\phi^{(m,n)\,\mu(m)}\psi^{(m,n)}_{\mu(m)}\,.
\end{equation}
Finally, the following condensed notations will be useful:
\begin{equation}
    \begin{array}{rlrl}
        g\phi^{(m,n)} := & q^{2}\frac{\partial^2}{\partial p^2}\phi^{(m,n)}\,, &
        g^{*}\phi^{(m,n)} := & p^2\partial_q^{2}\phi^{(m,n)}\,.
    \end{array}
\end{equation}

Ambient derivative along any direction decomposes into covariant AdS derivative and a radial complement $\partial_r$ (which acts algebraically on subspaces with a fixed radial weight). For example, symmetrised gradient and divergence for $\Phi\in\mathcal{S}_{\Delta}$ read respectively as
\begin{equation}\label{eq:diff_operators}
    P\cdot\partial_{X}\,\Phi = \left(p^{\mu}\nabla_{\mu} - \frac{q}{r}\,\Delta\right)\Phi,\quad \partial_{P}\cdot \partial_{X}\,\Phi = \left(\frac{\partial}{\partial p^{\mu}}\nabla^{\mu} + \frac{\Delta}{r}\,\frac{\partial}{\partial q}\right)\Phi
\end{equation}
With the aid of relations \eqref{eq:ambient_covariant_derivative} and \eqref{eq:cov_D_pq}, operators \eqref{eq:diff_operators} are rewritten in terms of AdS tensor fields: for $\Phi^{(m+n)}_{n}$ in \eqref{eq:radial_decomposition_AdS} one has
\begin{equation}\label{eq:Pderiv-AdS}
\def\arraystretch{2}
\begin{array}{rl}
    \displaystyle P\cdot \partial_{X}\Phi^{(m+n)}_{n}(X|P) = & \displaystyle \left(p^{\mu} \nabla_{\mu} - q\,\frac{\Delta + m}{r} - \frac{p^2}{r}\,\frac{\partial}{\partial q} \right)\,\phi^{(m,n)}\,,\\
    \displaystyle \partial_ {P}\cdot\partial_{X}\,\Phi^{(m+n)}_{n}(X|P) = & \displaystyle \left(\frac{\partial}{\partial p_{\mu}}\nabla_{\mu}  -\frac{d+1+m - \Delta}{r}\,\frac{\partial}{\partial q} - \frac{q}{r}\,\frac{\partial^2}{\partial p^2}\right)\,\phi^{(m,n)}\,,
\end{array}
\end{equation}
where, in the above expressions, covariant AdS derivative acts only on the AdS field while $p^{\mu},q$ are treated as independent variables. In order to read off the action of $\hO$ in terms of AdS fields we supplement the above expressions by that for the ambient d'Alembertian:
\begin{multline}\label{eq:dAlemb-AdS}
    \Box\Phi^{(m+n)}_{n}(X|P) = \left(\nabla^{2} - \frac{2}{r}\left(q\,\frac{\partial}{\partial p_{\mu}}\nabla_{\mu} + p^{\mu}\nabla_{\mu}\,\frac{\partial}{\partial q}\right) + \frac{\Delta(d-\Delta)}{r^2} \right.\\ \left. +\frac{1}{r^2}\left(q\,\frac{\partial}{\partial p_{\lambda}} + p^{\lambda}\frac{\partial}{\partial q}\right)\left(q\,\frac{\partial}{\partial p^{\lambda}} + p_{\lambda}\frac{\partial}{\partial q}\right) \right)\,\phi^{(m,n)}\,.
\end{multline}

A spin-$s$ field $\Phi^{(s)}\in\hS^{(s)}_{\Delta}$ has the following ghost-degree expansion 
\begin{equation}
    \Phi^{(s)} = \left.\Phi^{(s)}\right|_{\mathrm{gh} = -1} + \left.\Phi^{(s)}\right|_{\mathrm{gh} = 0} + \left.\Phi^{(s)}\right|_{\mathrm{gh} = +1} + \left.\Phi^{(s)}\right|_{\mathrm{gh} = +2}
\end{equation}
with particular components written as
\begin{equation}\label{eq:ghost_expansion}
    \begin{array}{rl}
        \left.\Phi^{(s)}\right|_{\mathrm{gh} = -1} = & \displaystyle b\,\sum_{m+n = s-1} \varepsilon^{(m,n)}\,,\\
        \left.\Phi^{(s)}\right|_{\mathrm{gh} = 0} = & \displaystyle \sum_{m+n = s} B^{(m,n)} + c_0 b \,\sum_{m+n = s-1} C^{(m,n)} + cb\,\sum_{m+n = s-2} D^{(m,n)}\,,\\
        \left.\Phi^{(s)}\right|_{\mathrm{gh} = +1} = & \displaystyle c_0\,\sum_{m+n = s}\widetilde{B}^{(m,n)} + c\,\sum_{m+n = s-1}\widetilde{C}^{(m,n)} + c_0cb\,\sum_{m+n = s-2}\widetilde{D}^{(m,n)}\,,\\
        \left.\Phi^{(s)}\right|_{\mathrm{gh} = +2} = & \displaystyle c_0c\,\sum_{m+n = s-1} F^{(m,n)}\,,
    \end{array}
\end{equation}
where the field components on the right-hand side do not depend on the ghost variables. The introduced ingredients are sufficient for reformulating ambient expressions in purely AdS terms.

\section{Proofs}\label{sec:proofs}

\subsection{Proof of Lemma \bref{lem:radial_decomposition}}\label{sec:proof_lem_1}
    It is sufficient to check the statement for the monomials $\Phi(X|P)$ of homogeneity degree $s$ in $P_A$: $(P\cdot\partial_P - s)\,\Phi(X|P) = 0$. Consider the first-order PDE in $P_A$
    \begin{equation}\label{eq:radial_PDE}
        \big((T\cdot P)(T\cdot \partial_P) - n\,T^2\big) f(X|P) = 0
    \end{equation}
    in a neighbourhood where $T\cdot P\neq 0$. Its general solution is expressed as $f(X|P) = (T\cdot P)^n\,\widetilde{f}(X|P)$ such that $\widetilde{f}(X|P)$ is a general solution of
    \begin{equation}
        (T\cdot \partial_P) \widetilde{f}(X|P) = 0\,.
    \end{equation}
    The solutions of the latter equation form a commutative algebra under the point-wise product generated by $X^A$ and the $d+1$ independent characteristics among
    \begin{equation}
        \widetilde{P}^A = P^A - \frac{(T\cdot P)}{T^2}\,T^A\,.
    \end{equation}
    Note that $T\cdot \widetilde{P} = 0$, so only $d+1$ components are indeed independent. Taking this into account and applying the decomposition $P^A = \widetilde{P}^A + \frac{(T\cdot P)}{T^2}\,T^A$ to the field $\Phi(X|P)$, one finds that for each integer $ n\in\{0,\dots,s\}$ there is a unique component $\Phi_n(X|P)$ satisfying \eqref{eq:radial_PDE}\,.

\subsection{Proof of Lemma \bref{lem:T1_resolve}}\label{sec:proof_lem_2}
    Because the operator $U$ \eqref{eq:operator_U} is invertible, for any $\Phi\in\hS$ there exists the field $\widetilde{\Phi} = U^{-1}\Phi$. In this respect, for $\Phi = U\widetilde{\Phi}$ one has (see \eqref{eq:resolution_T1})
    \begin{equation}
        (\sfT_1)^t\, \Phi = U\,(X\cdot\partial_P)^t\, \widetilde{\Phi}\,.
    \end{equation}
    Note that for the non-zero components $\widetilde{\Phi}_n$ of the radial decomposition of $\widetilde{\Phi}$ one has $(X\cdot \partial_P)^{n+j}\widetilde{\Phi}_n = 0$ only for $j\geqslant 1$ (see the proof of the Lemma \bref{lem:radial_decomposition}, Section \bref{sec:proof_lem_1}). Therefore satisfying $(\sfT_1)^t \Phi = 0$ is equivalent to having $\widetilde{\Phi}_{t+j} = 0$ for all integer $j\geqslant 0$ in the radial decomposition \eqref{eq:radial_decomposition}. 

\subsection{Proof of Lemma \bref{lem:conjugation}}\label{sec:proof_lem_3}
    Any local ambient operator is a polynomial in the oscillators and ambient derivatives. Conjugation $(\cdot)^{\dagger}$ preserves $h$-degree (see the comment after \eqref{eq:braket_h-degree}), therefore if $\op$ is an ambient operator of $h$-degree $0$ then so is $\op^{\#}$. Rules \eqref{eq:dagger}, \eqref{eq:dagger_ext} imply that
    \begin{equation}\label{eq:pre_conjugation}
        \big(\Psi,\op\Phi\big)_K = (-)^{\gh(\Psi)\cdot\gh(\op)}  \big(\op^{\#}\Psi,\Phi\big)_K + \partial_A J^A_{\Psi,\Phi}\,\vol_{X}
    \end{equation}
    for a certain $J^C_{\Psi,\Phi}$ such that $(X\cdot\partial + d+2)\partial_C J^C_{\Psi,\Phi} = 0$. The latter is equivalent to $(X\cdot\partial + d+1) J^C_{\Psi,\Phi} = 0$. The expression \eqref{eq:pre_conjugation} has the form \eqref{eq:conjugate_property} thanks to the following identity for any $\mathcal{J} = i_{J}\,\vol_{X}$ with $J = J^C \partial_C$:
    \begin{equation}
        \partial_C J^C\,\vol_{X} = \dd\mathcal{J} + (X\cdot\partial + d + 1)J^C\,i_{\partial_C} \vol_{X} \,.
    \end{equation}
    \vskip 4 pt

    To show that the inner product $\big(\cdot,\cdot\big)_K$ is non-degenerate on $\hS_{\Delta}$, note that the constraint in \eqref{eq:subspaces} is a first-order linear PDE, so any element of $\hS_{\Delta}$ is uniquely defined by its unconstrained restriction to $AdS_{d+1}$ with $r = \ell$. Since the inner product in question is point-wise with respect to $\mathbb{R}^{2,d}_{+}$, one specifies attention to $r = \ell$, where non-degeneracy of $\big(\cdot,\cdot\big)_K$ on $\hS_{\Delta}$ is equivalent to non-degeneracy of $\langle\cdot,\cdot\rangle$ on $\hS$.
    \vskip 4 pt
    
    Similarly, one proves the uniqueness of the operator $\op^{\#}$.  
    Suppose that both $\op^{\prime}_1$, $\op^{\prime}_2$ are conjugate to $\op$ in the sense that \eqref{eq:conjugate_property} holds. Then $\omega = \op^{\prime}_1 - \op^{\prime}_2$ satisfies
    \begin{equation}\label{eq:difference_exact}
        \big(\omega\Psi,\Phi\big)_K = \dd \mathcal{J}_{\Psi,\Phi}
    \end{equation}
    for some local $d$-form $\mathcal{J}_{\Psi,\Phi}$. Pullback of \eqref{eq:difference_exact} to $AdS_{d+1}$ for $r = \ell$ and integration over any neighbourhood with a choice of $\Phi$ with finitely supported pullback to $AdS_{d+1}$ gives $0$ for any $\Psi\in\hS_{\Delta}$. Since $\Psi\in\hS_{\Delta}$ is uniquely reconstructed from its pullback to $AdS_{d+1}$, and since $\big(\cdot,\cdot\big)_K$ is point-wise with respect to $\mathbb{R}^{2,d}_{+}$ and non-degenerate on $\hS_{\Delta}$, one concludes that restriction of $\omega$ to $\hS_{\Delta}$ is trivial.

\subsection{Proof of Lemma \bref{lem:deformation}}\label{sec:proof_K}

For the proof of non-degeneracy of $\big(\cdot,\cdot\big)_{K}$ on $\hT_{\Delta}$ for any $\Delta \in \mathbb{R}$ we assume that $K$ is an invertible transformation of $\hT_{\Delta}$, which will be proven to be true in the sequel. In this case, the pairing $\langle K\cdot,\cdot\rangle$ on $\hT_{\Delta}$ is non-degenerate iff so is the pairing $\langle \cdot,\cdot\rangle$ on $\hT_{\Delta}$. To check that $\langle \cdot,\cdot\rangle$ is non-degenerate on $\hT_{\Delta}$, note that by Lemma \bref{lem:T1_resolve} the problem is equivalent to checking non-degeneracy of $\langle  \cdot,\cdot\rangle$ on the space of fields with constrained radial decomposition and which are subject to the constraint 
\begin{equation}\label{eq:T2_modifid}
    U^{-1}\sfT_{2} U = \sfT_{2} - \frac{4}{X^2} (X\cdot \partial_{P}) \,c_{0}\frac{\partial}{\partial c}\,.
\end{equation}
The latter check can be done explicitly in terms of decompositions described in Appendix \bref{sec:formulae}. 
\vskip 4 pt

For the sequel, we will make use of the following simple lemma.

\begin{lem}\label{lem:technical_constraints}
Let $V$ be a vector space with a  bilinear form $G$ and $\{\varphi_{\alpha}\}_{\alpha = 1,\dots,r}$ some independent linear forms on $V$. Let $V_0\subset V$ be the common kernel of the linear forms $\{\varphi_{\alpha}\}$. Then restriction of $G$ to $V_0$ is symmetric (respectively, anti-symmetric) iff there exist linear forms $c_\alpha$ on $V$ such that for $\sigma = 0$ (respectively, $\sigma = 1$) holds
\begin{equation}
    G(v,w) - (-)^{\sigma} G(w,v) = \sum_{\alpha=1}^{r} \big(c_{\alpha}(v)\varphi_{\alpha}(w) - (-)^{\sigma} c_{\alpha}(w)\varphi_{\alpha}(v)\big)\,.
\end{equation}
\end{lem}
\begin{proof}
The fact that restriction of $G$ to the subspace $V_0$ is symmetric (respectively, anti-symmetric) is equivalent to saying that its anti-symmetric (respectively, symmetric) part is proportional to $\varphi_{\alpha}$. For a fixed $\sigma = 0,1$ consider the following projectors
\begin{equation}
    \Pi^{(\sigma)}_{\epsilon} G(v,w) = \frac{1}{2}\left(G(v,w) + (-)^{\sigma+\epsilon}G(w,v)\right),\quad \epsilon = 0,1\,
\end{equation}
satisfying $\Pi^{(\sigma)}_{0}\Pi^{(\sigma)}_{1} = \Pi^{(\sigma)}_{1}\Pi^{(\sigma)}_{0} = 0$. Anti-symmetric (respectively, symmetric) part of the metric is given by $\Pi^{(\sigma)}_{1}G$ for $\sigma = 0$ (respectively, $\sigma = 1$), and hence the condition on symmetricity (respectively, anti-symmetricity) of $G$ on $V_0$ is formulated as follows:
\begin{equation}
    \Pi^{(\sigma)}_1 G(v,w) = \sum_{\alpha}\big(L_{\alpha}(v)\varphi_{\alpha}(w) + R_{\alpha}(w)\varphi_{\alpha}(v)\big)\,.
\end{equation}
Applying $\Pi^{(\sigma)}_0$ to the both sides of the above expression yields the consistency condition
\begin{equation}
    L_{\alpha} + (-)^\sigma R_{\alpha} = 0
\end{equation}
and thus leads to the expression stated in the assertion.
\end{proof}

Fixing the operator $K = r^{-(d-2\Delta)}\,\big(U^{\dagger}U^{-1}\big)^{-1}\mathfrak{D}U^{\dagger}U^{-1}$ on $\hT_{\Delta}$ consists in fixing $\mathfrak{D}$. For that purpose we consider the condition coming from the requirement that, for any spin $s\geqslant 1$ in the ghost-$\pm 1$ sectors,
\begin{equation}
    \langle K\Psi^{(s)},\hO\Phi^{(s)}\rangle + \langle K\Phi^{(s)},\hO\Psi^{(s)}\rangle = \nabla_{\lambda}j^{\lambda}\,,
\end{equation}
with the special choice
\begin{equation}
\def\arraystretch{1.4}
    \begin{array}{rl}
        \Psi^{(s)} = & b\,\sum_{m+n = s-1} \varepsilon^{\prime(m,n)} + c\,\sum_{m+n = s-1}\widetilde{C}^{\prime(m,n)} \,,  \\
        \Phi^{(s)} = & b\,\sum_{m+n = s-1} \varepsilon^{(m,n)} + c\,\sum_{m+n = s-1}\widetilde{C}^{(m,n)} \,,
    \end{array}
\end{equation}
where the component fields do not depend on the ghost variables (see \eqref{eq:ghost_expansion} and Appendices \bref{sec:coordinates}, \bref{sec:formulae} for details). In particular, we will be interested only in mass-like terms of the form $\varepsilon\widetilde{C}$ ({\it i.e.} the ones without AdS covariant derivatives).
\vskip 4pt

To compute the action of $K$ defined by \eqref{eq:deformation_K_ansatz}, we express the operators $U$ and $U^{\dagger}$ \eqref{eq:operator_U} in terms of the radial oscillator $q$ \eqref{eq:AdS_ocsillators}:
\begin{equation}
    U = 1 - \frac{2q}{r}\,c_0\frac{\partial}{\partial c},\quad U^{\dagger} = 1 - \frac{2}{r}\,\partial_{q}\,c_0 b\qquad\Rightarrow\qquad U^{\dagger}U^{-1}=1+ \frac{2}{r}\,c_0\,\Big(q\frac{\partial}{\partial c}-b\,\partial_{q}\Big)\,.
\end{equation}
Then one obtains (throughout this section we will write $w^{*} = d - 2\Delta$ for brevity)
\begin{equation}
\def\arraystretch{1.4}
\begin{array}{rl}
    r^{w^{*}}K\Phi^{(s)} = & \displaystyle b\sum_{m+n=s-1} \nu^{(s|n)}_{\Delta}\varepsilon^{(m,n)} \\
     & + \,\displaystyle c_0\sum_{m+n = s} \big(\nu^{(s|n)}_{\Delta} - \nu^{(s|n-1)}_{\Delta}\big)\,\frac{2q}{r}\,\widetilde{C}^{(m,n-1)}\\
     & + \,\displaystyle c\sum_{m+n = s-1} \nu^{(s|n)}_{\Delta} \widetilde{C}^{(m,n)}\\
     & + \,\displaystyle c_0cb\sum_{m+n = s-2}\big(\nu^{(s|n)}_{\Delta} - \nu^{(s|n+1)}_{\Delta}\big)\,\frac{2}{r}\,\partial_{q}\widetilde{C}^{(m,n+1)}\,.
\end{array}
\end{equation}
Performing ghost pairings leads to
\begin{equation}\small
\def\arraystretch{1.4}
\begin{array}{rl}
    r^{w^{*}}\langle K\Psi^{(s)},\hO\Phi^{(s)}\rangle = & \displaystyle -\sum_{k+l = s-1}\,\sum_{m+n = s-1}\nu^{(s|l)}_{\Delta}\langle \widetilde{C}^{\prime(k,l)},\Box\varepsilon^{(m,n)}\rangle_{\text{Fock}} \\
    & +\, \displaystyle\sum_{k+l = s-1}\,\sum_{m+n  = s-1} \nu^{(s|l)}_{\Delta} \langle \varepsilon^{\prime(k,l)},\Box\widetilde{C}^{\prime(m,n)}\rangle_{\text{Fock}} \\
    & -\, \displaystyle\frac{2}{r}\sum_{k+l = s}\,\sum_{m+n = s-1} \big(\nu^{(s|l)}_{\Delta} - \nu^{(s|l-1)}_{\Delta}\big)\,\langle q\widetilde{C}^{\prime(k,l-1)},P\cdot\partial_X \varepsilon^{(m,n)} \rangle_{\text{Fock}}\\
    & -\, \displaystyle\frac{2}{r}\sum_{k+l = s-2}\,\sum_{m+n = s-1} \big(\nu^{(s|l)}_{\Delta} - \nu^{(s|l+1)}_{\Delta}\big)\,\langle\partial_{q}\widetilde{C}^{(k,l+1)},\partial_{P}\cdot\partial_{X}\varepsilon^{(m,n)} \rangle_{\text{Fock}}\,.
\end{array}
\end{equation}
By virtue of \eqref{eq:Pderiv-AdS}, \eqref{eq:dAlemb-AdS} one gets the following mass-like terms:
\begin{equation}\small
\def\arraystretch{1.4}
\begin{array}{rl}
    r^{w^{*}}\langle K\Psi^{(s)},\hO\Phi^{(s)}\rangle = & \\
    -  \displaystyle \frac{1}{r^2}\sum_{m+n = s-1}\nu^{(s|n)}_{\Delta} & \left((m + 2mn + (d+1)n + (\Delta - 1)(d - \Delta + 1))\,\widetilde{C}^{\prime(m,n)}\varepsilon^{(m,n)}\right.\\
    \hfill &\qquad  +\left. \widetilde{C}^{\prime(m,n)}\,g\varepsilon^{(m+2,n-2)} + \widetilde{C}^{\prime(m,n)}\,g^* \varepsilon^{(m-2,n+2)}\right)\\
    + \displaystyle \frac{1}{r^2}\sum_{m+n = s-1}\nu^{(s|n)}_{\Delta} & \left((m+2mn+(d+1)n +(\Delta+1)(d-\Delta-1))\,\varepsilon^{\prime(m,n)}\widetilde{C}^{(m,n)} \right.\\
    \hfill &\qquad \left. + \varepsilon^{\prime(m,n)}\,g\widetilde{C}^{(m+2,n-2)} + \varepsilon^{\prime(m,n)}\,g^*\widetilde{C}^{(m-2,n+2)}\right)\\
    \displaystyle - \frac{2}{r^2}\sum_{m+n = s-1} (\nu^{(s|n+1)}_{\Delta} - \nu^{(s|n)}_{\Delta}) & \left((n+1)(\Delta - 1 + m)\, \widetilde{C}^{\prime(m,n)}\varepsilon^{(m,n)} + \widetilde{C}^{\prime(m,n)}\,g^* \varepsilon^{(m-2,n+2)}\right)\\
    \displaystyle + \frac{2}{r^2}\sum_{m+n = s-1} (\nu^{(s|n)}_{\Delta} - \nu^{(s|n-1)}_{\Delta}) &  \left(n(d-\Delta+m+2)\,\widetilde{C}^{\prime(m,n)}\varepsilon^{(m,n)} + \widetilde{C}^{\prime(m,n)}\,g\varepsilon^{(m+2,n-2)} \right)\\
     +\,\dots
\end{array}
\end{equation}
with all terms with derivatives are not written explicitly (they are denoted by the ellipses). We look for an operator $K$ such that the pairing $\langle K\Psi^{(s)},\hO\Phi^{(s)}\rangle$ is anti-symmetric, so we concentrate on its symmetric part and proceed by fixing it to zero:
\begin{equation}\label{eq:symmetric_zero}\small
\def\arraystretch{1.4}
\begin{array}{rll}
    \multicolumn{3}{l}{\displaystyle  
    r^{d-2\Delta}\left(\langle K\Psi^{(s)},\hO\Phi^{(s)}\rangle + \langle K\Phi^{(s)},\hO\Psi^{(s)}\rangle\right)}\\
    \hfill & \displaystyle =\,\frac{2}{r^2}\sum_{m+n = s-1} & \left(\nu^{(s|n)}_{\Delta}(d-2\Delta) + n\, (\nu^{(s|n)}_{\Delta}-\nu^{(s|n-1)}_{\Delta})(d-\Delta+m+2)\right. \hfill \\  
    \hfill & \hfill & \left. - (n+1)\,(\nu^{(s|n+1)}_{\Delta} - \nu^{(s|n)}_{\Delta})(\Delta-1 +m)\right)\,\left(\varepsilon^{\prime(m,n)}\widetilde{C}^{(m,n)} + \varepsilon^{(m,n)}\widetilde{C}^{\prime(m,n)}\right)\\
    \hfill & \multicolumn{2}{l}{\displaystyle +\,\frac{1}{r^2}\sum_{m+n = s-1}\big(\nu^{(s|n)}_{\Delta} - 2\nu^{(s|n-1)}_{\Delta} + \nu^{(s|n-2)}_{\Delta}\big)\big(C^{\prime(m,n)}g\varepsilon^{(m+2,n-2)} + C^{(m,n)}g\varepsilon^{\prime(m+2,n-2)}\big)} \\
    \hfill & \multicolumn{2}{l}{\displaystyle +\,\frac{1}{r^2}\sum_{m+n = s-1}\big(\nu^{(s|n)}_{\Delta} - 2\nu^{(s|n-1)}_{\Delta} + \nu^{(s|n-2)}_{\Delta}\big)\big(\varepsilon^{\prime(m,n)}gC^{(m+2,n-2)} + \varepsilon^{(m,n)}gC^{\prime(m+2,n-2)}\big)
    +\dots}
\end{array}
\end{equation}
Recall that the field components are not independent: by definition of $\hT^{(s)}_{\Delta}$ in \eqref{eq:ghost_traceless} the trace constraint in terms of the AdS fields (with $n\geqslant 2$) reads
\begin{equation}\label{eq:Appendix_trace_constraints}
    \varphi^{(s|n)}_{1}(\Phi^{(s)}) = g\varepsilon^{(m+2,n-2)} - n(n-1)\, \varepsilon^{(m,n)} = 0,\quad \varphi^{(s|n)}_{2}(\Phi^{(s)}) = g\widetilde{C}^{(m+2,n-2)} - n(n-1)\,\widetilde{C}^{(m,n)} = 0\,.
\end{equation}

For the lowest values $n = 0,1$, the traceless parts of the field components are free from the constraints \eqref{eq:Appendix_trace_constraints} and hence for the coefficients in front of $\varepsilon^{(s-1,0)}\widetilde{C}^{(s-1,0)}$ and $\varepsilon^{(s-2,1)}\widetilde{C}^{(s-2,1)}$ one directly imposes 
\begin{equation}
\def\arraystretch{1.5}
\begin{array}{ll}
    \nu^{(s|0)}_{\Delta} (d-2\Delta) - (\nu^{(s|1)}_{\Delta} - \nu^{(s|0)}_{\Delta})(\Delta + s - 2) = 0 & (\text{for}\quad n = 0)\,,\\
    \nu^{(s|1)}_{\Delta} (d-2\Delta) + (\nu^{(s|1)}_{\Delta} - \nu^{(s|0)}_{\Delta}) (d - \Delta + s) - 2(\nu^{(s|2)}_{\Delta} - \nu^{(s|1)}_{\Delta}) (\Delta + s - 3) = 0 & (\text{for}\quad n = 1)\,.
\end{array}
\end{equation}
Normalising the first coefficient $\nu^{(s|0)}$ to unity brings us to
\begin{equation}\label{eq:Appendix_D_initial}
\nu^{(s|0)}_{\Delta} = 1 \,,\quad   \nu^{(s|1)}_{\Delta} = \frac{[\bar{\Delta} + s - 2]_1}{[\Delta + s - 2]_1}\,, \quad \nu^{(s|2)}_{\Delta} = \frac{[\bar{\Delta} + s - 2]_2 + w^{*}}{[\Delta + s - 2]_2}\,,
\end{equation}
where $\bar{\Delta} = d - \Delta$ and $[x]_k = x(x-1)\dots(x-k+1)$ denotes the falling Pochhammer symbol. 
\vskip 4 pt

For $n\geqslant 2$ the field components  are subject to the relations \eqref{eq:Appendix_trace_constraints}. The latter are algebraic and hence can be treated point-wise, what brings us to the situation described in Lemma \bref{lem:technical_constraints} with the two linear forms being respectively \eqref{eq:Appendix_trace_constraints}. In more detail, proceeding step by step for all admissible $n \geqslant 2$, for each fixed $n$ (and $m = s - 1 - n$) one considers the field components $g\varepsilon^{(s-n+1,n-2)}$, $\varepsilon^{(s-n-1,n)}$, $g\widetilde{C}^{(s-n+1,n-2)}$, $\widetilde{C}^{(s-n-1,n)}$. Anti-symmetry of $\langle K\Psi^{(s)},\Omega\Phi^{(s)}\rangle$ on the surface of the constraints \eqref{eq:Appendix_trace_constraints} amounts to finding two combinations
\begin{equation}
    c^{(s|n)}_{\alpha}(\Phi^{(s)}) = c_{\alpha\,1}\, g\varepsilon^{(s-n+1,n-2)} + c_{\alpha\,2}\, \varepsilon^{(s-n-1,n)} + c_{\alpha\,3}\, g\widetilde{C}^{(s-n+1,n-2)} + c_{\alpha\,4}\, \widetilde{C}^{(s-n-1,n)}\,,\quad (\alpha = 1,2)\,,
\end{equation}
such that
\begin{multline}
    r^{w^{*}}\big(\langle K\Psi^{(s)},\hO\Phi^{(s)}\rangle + \langle K\Phi^{(s)},\hO\Psi^{(s)}\rangle\big)\\ = \sum_{\alpha = 1,2}\big(\langle c^{(s|n)}_{\alpha}(\Psi^{(s)}),\varphi^{(s|n)}_{\alpha}(\Phi^{(s)})\rangle_{\text{Fock}} + \langle c^{(s|n)}_{\alpha}(\Phi^{(s)}),\varphi^{(s|n)}_{\alpha}(\Psi^{(s)})\rangle_{\text{Fock}}\big)\,.
\end{multline}
Substituting \eqref{eq:symmetric_zero} to the above requirement immediately allows one to fix $c_{11} = c_{12} = c_{23} = c_{24} = 0$. For the terms of the form $\varepsilon\widetilde{C}$ one arrives at
\begin{multline}\label{eq:Appendix_antysim_traceless}
    2(d-2\Delta) - 2(n+1)(\Delta + s - 2 - n)(\nu^{(s|n+1)}_{\Delta}-\nu^{(s|n)}_{\Delta}) - 2n\,(d-\Delta + s + 1 - n) \\
    = - n(n-1)(c_{14} + c_{22})\,.
\end{multline}
For the terms of the form $\varepsilon\,g\widetilde{C}$ and $g\varepsilon\,\widetilde{C}$ one gets, respectively, (recall \eqref{eq:field_components_Fock})
\begin{equation}\label{eq:Appendix_antisym_traces}
    \nu^{(s|n)}_{\Delta} - 2 \nu^{(s|n-1)}_{\Delta} + \nu^{(s|n-2)}_{\Delta} = -c_{13} + c_{22} = -c_{21} + c_{14}\,.
\end{equation}
Finally, for the terms $g\varepsilon\, g\widetilde{C}$ one gets $c_{13} + c_{21} = 0$. Substituting the latter equation to \eqref{eq:Appendix_antisym_traces} one derives
\begin{equation}
    2\,\big(\nu^{(s|n)}_{\Delta} - 2 \nu^{(s|n-1)}_{\Delta} + \nu^{(s|n-2)}_{\Delta}\big) = c_{14} + c_{22}\,,
\end{equation}
One substitutes the above relation to \eqref{eq:Appendix_antysim_traceless}, which leads to the following recurrence relation for $\nu^{(s|n)}_{\Delta}$:
\begin{equation}\label{eq:Appendix_recursion_1}
\def\arraystretch{1.4}
    \begin{array}{l}
        nf^{(s|n)}_{\Delta} - (n-1)f^{(s|n-1)}_{\Delta} = 0\quad\text{for}\quad n\geqslant 3\,,\\
        \text{with} \quad f^{(s|2)}_{\Delta} = 0\,,
    \end{array}
\end{equation}
where 
\begin{equation}\label{eq:def_fnu}
    f^{(s|n)}_{\Delta} = \big((\Delta + s - 2) - (n-1)\big)\,\nu^{(s|n)}_{\Delta} - (\bar{\Delta} + s - 2)\,\nu^{(s|n-1)}_{\Delta} + (n-1)\,\nu^{(s|n-2)}_{\Delta}\,.
\end{equation}
The initial data for \eqref{eq:Appendix_recursion_1} follows from \eqref{eq:Appendix_D_initial}. As a consequence, one deduces
\begin{equation}
    f^{(s|n)}_{\Delta} = 0\quad \text{for all}\quad n\geqslant 2\,.
\end{equation}
By recalling \eqref{eq:def_fnu}, the above equation is a recursion relation for $\nu^{(s|n)}_{\Delta}$. We look for a solution in the form 
\begin{equation}
    \nu^{(s|n)}_{\Delta} = \frac{q^{(s|n)}_{\Delta}}{[\Delta+s-2]_n}\,,
\end{equation}
which results in the following recursion relation for $q^{(s|n)}_{\Delta}$:
\begin{equation}
    q^{(s|n+2)}_{\Delta} - (\bar{\Delta}+s-2)\,q^{(s|n+1)}_{\Delta} + (n+1)(\Delta + s - (n+2))\,q^{(s|n)}_{\Delta} = 0\,.
\end{equation}
Assuming a generating function $Q^{(s)}_{\Delta}(t) = \sum_{n = 0}^{\infty} q^{(s|n)}_{\Delta}\frac{t^n}{n!}$ one arrives at the following Cauchy problem:
\begin{equation}
\def\arraystretch{1.5}
    \begin{array}{ll}
        \multicolumn{2}{l}{\displaystyle \big(1-t^2\big)\,Q^{(s)\prime\prime}_{\Delta}(t) - \big(\bar{\Delta}+s-2-(\Delta+s-4)t\big)\,Q^{(s)\prime}_{\Delta}(t) + \big(\Delta+s-2\big)\,Q^{(s)}_{\Delta}(t) = 0,}\\ 
        \text{with the initial data} & Q^{(s)}_{\Delta}(0) = 1,\quad Q^{(s)\prime}_{\Delta}(0) = \bar{\Delta} + s - 2\,.
    \end{array}
\end{equation}
The solution to the above equation is
\begin{equation}
    Q^{(s)}_{\Delta}(t) = (1-t)^{-(\frac{d}{2}-\Delta)}(1+t)^{\frac{d}{2}+s-2}\,.
\end{equation}
Finally, with the aid of the above generating function one finally expresses the sought eigenvalues of $\mathfrak{D}$ in terms of the Euler hypergeometric function:
\begin{equation}
    \nu^{(s|n)}_{\Delta} = \frac{\big[\frac{d}{2}+s-2\big]_n}{\big[\Delta + s - 2\big]_n}\;{}_2F_{1}\left(\begin{array}{c} -n \quad\quad \frac{d}{2}-\Delta\\ \frac{d}{2} + s - 1 - n\end{array}\,;\,-1\right)\,.
\end{equation}

To complete the proof, it is left to verify that for the so-constructed operator $K$ the inner product $\big(\cdot,\cdot\big)_{K}$ is BRST-anti-invariant in all ghost sectors. As a result of a lengthy but straightforward computation, for any ghost-extended ambient field $\Phi^{(s)}$ of spin $s$ holds
\begin{equation}\label{eq:K_commutator}
\def\arraystretch{1.4}
\begin{array}{rlcl}
    [\hO,K]\,\Phi^{(s)} = & \frac{1}{r}\,c\,[\mathfrak{D},q]\,\sfT_2\Phi^{(s)} & - & \sfT_2^{\dagger}\,\frac{1}{r}\frac{\partial}{\partial b}\,[\mathfrak{D},\partial_q]\,\Phi^{(s)} \\
    & +\, c_0\frac{1}{r^2}\,\big[[\mathfrak{D},q],q\big]\,\frac{\partial}{\partial c}c\,\sfT_2\Phi^{(s)} & - & \sfT_2^{\dagger}c_0\,\frac{1}{r^2}\,\big[[\mathfrak{D},\partial_q],\partial_q\big]\,b\frac{\partial}{\partial b}\Phi^{(s)}\\
    & -\, c_0\,\frac{1}{r^2}\,\big[[\mathfrak{D},q],q\big]\,c\frac{\partial}{\partial c}\,\sfT_2\Phi^{(s)} & + & \sfT_2^{\dagger}c_0\,\frac{1}{r^2}\,\big[[\mathfrak{D},\partial_q],\partial_q\big]\,b\frac{\partial}{\partial b}\Phi^{(s)}\\
    & -\, c_0\,\frac{1}{r^2}\,\big[[\mathfrak{D},q],\partial_q\big]\,2bc\,\sfT_2\Phi^{(s)} & + & \sfT^{\dagger}_2c_0\,\frac{1}{r^2}\,\big[[\mathfrak{D},q\big],\partial_q]\,2\frac{\partial}{\partial b}\frac{\partial}{\partial c}\Phi^{(s)}\,.
\end{array}
\end{equation}
For any $\Phi^{(s)}\in\hT_{\Delta}$ the {\it rhs} of \eqref{eq:K_commutator} leads to $\langle[\Omega,K]\Phi^{(s)},\cdot\rangle = 0$.

\subsection{Proof of Theorems \bref{thm:massless}, \bref{thm:massive}}\label{sec:proof_thms}

    We have to check that ambient Lagrangians of the form $\mathbb{L}[\Phi] = \big(\Phi,\Omega\Phi\big)$, with 
    $\Phi\in\hT^{(s)}_{\Delta}$ and $\mathrm{gh}(\Phi) = 0$, lead to the equations of motion $\Omega\Phi = 0$. To start, one can apply a formal variation:
    \begin{equation}\label{deltaL}
        \delta \mathbb{L}[\Phi] = \big(\delta\Phi,\Omega\Phi\big) + \big(\Phi,\Omega\,\delta\Phi\big) = 2\big(\delta\Phi,\Omega\Phi\big) + \dd\mathcal{J}\,,
    \end{equation}
    where $\delta\Phi\in\hT_{\Delta}^{(s)}$ and, in the second equality, we have used the BRST-anti-invariance of $\big(\cdot,\cdot\big)$ on $\hT^{(s)}_{\Delta}$.
    A corollary of Theorem \ref{thm:EL_derivative} is that performing a formal variation either before or after restricting to $\hT^{(s)}_{\Delta}$ implies the same result $\Omega\Phi = 0$ with $\Phi\in\hT^{(s)}_{\Delta}$ (see, {\it e.g.}, Section \ref{sec:examples} for an explicit analysis of the massless case). Therefore putting \eqref{deltaL} to zero implies the correct equation of motion, as was to be shown. 
    \vskip 4 pt
    
    Pullback of $L[\Phi]$ to $AdS_{d+1}$ can be performed explicitly with the aid of expressions presented in Appendix \bref{sec:formulae}. In the case of massless fields (for $\Delta = 2-s$) one arrives at the following known Lagrangian density for the massless spin-$s$ field \cite{Sagnotti:2003qa,Fotopoulos:2006ci,Fotopoulos:2008ka}:
    \begin{equation}
\def\arraystretch{1.4}
\begin{array}{rl}
    \ell^{d+2s-4}\,L[\Phi^{(s)}] & = \tfrac{1}{s!}\, B^{(s)}\cdot \big(\Box + (s + (2-s)(d-2+s))\big)B^{(s)}\\ 
    \hfill & + \tfrac{1}{(s-1)!}\left( B^{(s)} \nabla C^{(s-1)} - C^{(s-1)}\nabla\cdot B^{(s)}\right) - \tfrac{1}{(s-1)!} C^{(s-1)} C^{(s-1)}\\
    \hfill & - \tfrac{1}{(s-2)!} D^{(s-2)}(\Box + (6 - s(s+d-1))) D^{(s-2)}\\ 
    \hfill & + \tfrac{1}{(s-2)!} \left(D^{(s-2)} \nabla\cdot C^{(s-1)} - C^{(s-1)}\nabla D^{(s-2)}\right)\\
    \hfill & + \tfrac{1}{(s-1)!}\varepsilon^{(s-1)}\big(\Box + (s-1)(2-d-s)
    \big)\widetilde{C}^{(s-1)}\\
    \hfill & -\tfrac{1}{(s-1)!}\widetilde{C}^{(s-1)}\big(\Box + (s-1)(2-d-s)\big)\varepsilon^{(s-1)}\\
    \hfill & -\tfrac{1}{(s-2)!} \left(\widetilde{D}^{(s-2)}\nabla\cdot\varepsilon^{(s-1)} + \varepsilon^{(s-1)}\nabla\widetilde{D}^{(s-2)}\right)\\
    \hfill & -\tfrac{1}{(s-1)!}\left(\varepsilon^{(s-1)}\nabla\cdot\widetilde{B}^{(s)} + \widetilde{B}^{(s)}\nabla\varepsilon^{(s-1)}\right)\,.
\end{array}
\end{equation}

\subsection{Technical details and proofs for Section \bref{sec:jets}}\label{sec:jet_proofs}

\begin{lem}\label{lem:Lagrangian_T_exact}
    For any $\alpha\in \bigwedge^{(n-1,0)}_{T}J\E$, written as $\beta = i_{H}\vol_{T}$ with some total vector field $H = H^{A}D_{A}$, one has
    \begin{equation}
        \ddh^{\prime}\beta = \mathrm{div}\,H_{\perp}\,\vol_{T}\,.
    \end{equation}
    where $H_{\perp} = H + i_{H}\vartheta\,T$.
\end{lem}
\begin{proof}
    By virtue of $\LD_{T-Z}\beta = 0$ one has
\begin{equation}\label{eq:T-condition_exact_L}
    i_{[Z,H]} \vol_{T} = i_{[T,H]} \vol_{T} + \mathrm{div}\, T\wedge i_{H}\vol_{T}\,,
\end{equation}
Next, one computes
\begin{equation}\label{eq:dh-exact_div}
\def\arraystretch{1.4}
\begin{array}{rl}
    \ddh^{\prime}\beta = & \ddh i_{H}\vol_{T} + \vartheta\wedge \LD_{Z} i_{H}\vol_{T} = - \ddh i_{T} i_{H} \vol + \vartheta\wedge (i_{[T,H]} \vol_{T} + \mathrm{div}\, T\wedge i_{H}\vol_{T}) \\
    = & -\LD_{T} i_{H} \vol + i_{T} \ddh i_{H}\vol - i_{[T,H]}(\vartheta\wedge \vol_{T}) + (i_{[T,H]}\vartheta)\wedge \vol_{T} + \mathrm{div}\,T\,\vartheta\wedge i_{H}\vol_{T}\,,
\end{array}
\end{equation}
where one has used \eqref{eq:T-condition_exact_L} to replace $\LD_{Z} i_{H}\vol_{T} = i_{[Z,H]}\vol_{T}$. Note that $\LD_{T} i_{H}\vol = i_{[T,H]}\vol + \mathrm{div}\,T\,i_{H}\vol$, and $i_{T} \ddh i_{H}\vol = \mathrm{div}\,H\,\vol_{T}$, where one recalls \eqref{eq:divergence}. Also $\vartheta\wedge \vol_{T} = -\vol$, so one has
\begin{equation}\label{eq:dh-exact_div_1}
    \ddh^{\prime}\beta = \mathrm{div}\,H\,\vol_{T} + i_{[T,H]}\vartheta\, \vol_{T} + \mathrm{div}\,T\,( \vartheta\wedge i_{H}\vol_{T} - i_{H} \vol)\,.
\end{equation}
For the last term in the above equation one makes use of the following simple formula:
\begin{equation}
    0 = i_{H}i_{T}(\vartheta \wedge \vol) = \vartheta\wedge i_{H}\vol_{T} - i_{H}\vol - (i_{H}\vartheta)\,\vol_{T}\,,
\end{equation}
which allows one to rewrite \eqref{eq:dh-exact_div_1} as follows:
\begin{equation}\label{eq:dh-exact_div_2}
    \ddh^{\prime}\beta = \mathrm{div}\,H\,\vol_{T} + (i_{[T,H]}\vartheta)\, \vol_{T} + \mathrm{div}\,T\,(i_{H}\vartheta)\,\vol_{T}\,.
\end{equation}
Finally, by virtue of $\ddh \vartheta = 0$ and $i_{T}\vartheta = -1$ one has
\begin{equation}
    \LD_{T}\vartheta = 0\quad \Rightarrow\quad i_{[T,H]} \vartheta = T\,(i_{H}\vartheta)\,,
\end{equation}
which allows one to rewrite the last two terms in \eqref{eq:dh-exact_div_2} as $\mathrm{div}(i_{H}\vartheta\,T) \vol_{T}$.
\end{proof}

\begin{proof}[Proof of Lemma \bref{lem:EL_maps}]
    To prove \eqref{eq:EL_T-in-form} one needs to show, for any $\lambda \in \bigwedge_{T}^{(n,0)}J\E$, that one has $i_{T}\widehat{\delta}\lambda = 0$ and $\LD_{T - Z} \widehat{\delta}\lambda$. The former property follows directly from the definition \eqref{eq:EL_derivative}. The latter follows from the fact that $i_{T}$, $i_{T}^{-1}$ and $\delta$ commute with $\LD_{T-Z}$. Indeed, $\delta$ commutes with prolongations of vector fields which are projectable on $\E$ (see, {\it e.g.}, \cite[Corollary 2.13]{Anderson_bicomplex}). And $T-Z = \pr\,\bar{T}$, with $\bar{T}$ \eqref{eq:T_on_E} being a projectable vector field on $\E$.
\vskip 4 pt

One has $\widehat{\delta}(\ddh^{\prime} \beta) = 0$ for any $\beta\in \bigwedge_{T}^{(n-1,0)} J\E$ thanks to \eqref{eq:Lagrangian_h-exact}, which is annihilated by \eqref{eq:EL_derivative_explicit}.
\vskip 4 pt

To prove {\it (2)}, note that if a $T$-Lagrangian vanishes on $\I$, it can be written as a sum of terms $z^{\sfa}_{A(q)} H_{\sfa}^{A(q)}\wedge \vol_{T}$ over $\sfa$ and index sets $A(q)$ (for $q \geqslant 0$), with $H_{\sfa}^{A(q)}$ being some local functions on $J\E$. Due to linearity, let us consider a generic single term $\lambda = z^{\sfa}_{A(q)} H_{\sfa}^{A(q)}\wedge \vol_{T}$, and perform the necessary check manifestly. In the homogeneous coordinates \eqref{eq:homogeneous_coordinates} on the base $\hM$ the following useful formulae take place:
\begin{equation}\label{eq:useful_formulae_hom_coord}
    [D_{A(q)},T] = q w\, D_{A(q)}\quad\text{and}\quad  z^{\sfa}_{A(q)} = T^{B}u^{\sfa}_{BA(q)} + (\Delta^{\sfa} + qw)\,u^{\sfa}_{A(q)}\,,\quad q \geqslant 0\,.
\end{equation}
Note also the following relations which take place in any coordinates:
\begin{equation}\label{eq:useful_commutators}
    [Z,D_A] = 0\,,\; [\partial_{\sfa},D_A] = 0\quad \text{and}\quad [\partial^{A(q)}_{\sfa},D_B] = \delta^{A}_{B}\,\partial^{A(q-1)}_{\sfa}\,, \quad q \geqslant 1\,.
\end{equation}
As a consequence of \eqref{eq:useful_formulae_hom_coord} and \eqref{eq:useful_commutators} one obtains $Z z^{\sfa}_{A(q)} = T^{B}z^{\sfa}_{BA(q)} + (\Delta + q w) z^{\sfa}_{A(q)}$, which allows one to rewrite the condition $\LD_{T-Z}\lambda = 0$ as follows:
\begin{equation}\label{eq:vanishing_hom_coord}
    F_{\sfa}^{A(q)}\,z^{\sfa}_{A(q)} = 0\,,
\end{equation}
with 
\begin{equation}\label{eq:vanishing_hom_coord_1}
    F_{\sfa}^{A(q)} = \rho^{-1}\,D_B(\rho\,T^B H_{\sfa}^{A(q)}) - (\Delta + q w) H_{\sfa}^{A(q)} - ZH_{\sfa}^{A(q)}\,.
\end{equation}
Since $z^{\sfa}_{A(q)}$ are independent, \eqref{eq:vanishing_hom_coord} implies that $F_{\sfa}^{A(q)} = \sum_{k \geqslant 0} f_{\sfa\sfb}^{A(q),B(k)}\,z^{\sfb}_{B(k)}$ (with some local functions $f_{\sfa\sfb}^{A(q),B(k)}$). By applying the map \eqref{eq:EL_derivative_explicit} to $\lambda$, one gets
\begin{equation}
\def\arraystretch{1.4}
\begin{array}{rl}
    \widehat{\delta}\lambda = & \displaystyle \ddv u^{\sfa} \wedge \sum_{k\geqslant 0} (-)^{k} \rho^{-1} D_{A(k)}\left(\rho\,z^{\sfb}_{B(m)} \partial^{A(p)}_{\sfa} H^{B(m)}_{\sfb}\right) \wedge \vol_{T}\, + \\
    \hfill & \ddv u^{\sfa}\wedge (-)^{q+1} \rho^{-1} D_{B}D_{A(q)}\big(\rho\, T^{B} H_{\sfa}^{A(q)}\big)\wedge \vol_{T} \, + \\
    \hfill & \ddv u^{\sfa}\wedge (-)^{q} \rho^{-1} D_{A(q)}\big(\rho\,(\Delta + qw)\,H_{\sfa}^{A(q)}\big)\wedge \vol_T\,.
\end{array}
\end{equation}
Because $D_{A}z^{\sfa}_{B(q)} = z^{\sfa}_{AB(q)}$, each term in the first line of the above expression is proportional to $z^{\sfa}_{A(p)}$ with some $p \geqslant q$. According to \eqref{eq:vanishing_hom_coord_1}, the second and the third lines together give
\begin{equation}\label{eq:zero_locus_I}
   \ddv u^{\sfa} \wedge (-)^{q+1} \rho^{-1} D_{A(q)} (\rho\, G_{\sfa}^{A(q)}) \wedge \vol_{T}\,,\quad\text{where}\quad G_{\sfa}^{A(q)} = Z H_{\sfa}^{A(q)} + F_{\sfa}^{A(q)}\,.
\end{equation}
The sub-bundle $\I$ belongs to the zero-locus of $G_{\sfa}^{A(q)}$ because the latter is a combination of terms proportional to $z^{\sfa}_{A(p)}$. Since $D_{A}$ are tangent to $\I$, the former expression in \eqref{eq:zero_locus_I} belongs to the zero-locus of $G_{\sfa}^{A(p)}$ as well. 
As a result, $\widehat{\delta}\lambda$ indeed vanishes on $\I$.
\vskip 4 pt
\end{proof}

\section{Coefficients $\nu^{(s|n)}_{\Delta}$}\label{hypergeom}

In this section we will present a detailed consideration of a number of equivalent expressions for $\nu^{(s|n)}_{\Delta}$ \eqref{eq:coefficients_nu} and also present a generating function for them.
\vskip 4 pt

First, we recall that Euler's hypergeometric function can be represented as a power series:
\begin{equation}
    {}_2F_{1}\left(\begin{matrix} a,\;b\\ c \end{matrix};\,z\right) = \sum_{m=0}^{\infty} \frac{(a)_m (b)_m}{(c)_m}\,\frac{z^m}{m!}\,,\quad |z|<1\,,
\end{equation}
where $(x)_n = x(x+1)\dots (x+n-1)$ is the rising Pochhammer symbol. In the particular case when $a = -n$ (negative integer or zero) the above series terminates, and thus converges for any $z\in\mathbb{C}$. Fixing $z = -1$ and multiplying by $[c]_n$ (see \eqref{eq:coefficients_nu}) gives
\begin{equation}
    [c]_n\cdot{}_2F_{1}\left(\begin{matrix} -n,\;b\\ c \end{matrix};\,-1\right) = \sum_{m=0}^{n} \begin{pmatrix} n\\m\end{pmatrix}(b)_m[c]_{n-m}\,,
\end{equation}
where we have used that $(-n)_m = \frac{n!}{(n-m)!}$. The following formula makes calculation of \eqref{eq:coefficients_nu} straightforward:
\begin{equation}
    \nu^{(s|n)}_{\Delta} = \frac{\displaystyle \sum_{m=0}^n \begin{pmatrix} n\\m\end{pmatrix} (\tfrac{d}{2}-\Delta)_m\, [\tfrac{d}{2} + s - 2]_{n-m}}{[\Delta + s - 2]_n}\,.
\end{equation}
In particular, the above expression manifests singular behavior of $\nu^{(s|t+j)}_{s+t-1}$ for $j\geqslant 0$ in the case of partially massless regime. Note that singularities never occur in the Lagrangians in Theorem \bref{thm:massive} due to the definition of $\hT_{\Delta}$: imposing the tangency constraint \eqref{eq:ghost_tangent} puts to zero particular components in the radial decomposition \eqref{eq:radial_decomposition}, such that singular coefficients $\nu^{(s|n)}_{\Delta}$ never appear (see Lemma \bref{lem:T1_resolve}).
\vskip 4 pt

In order to pack the coefficients \eqref{eq:coefficients_nu} into a single generating function we make use of Euler's transformation
\begin{equation}
    {}_2F_{1}\left(\begin{matrix} a,\;b\\ c \end{matrix};\,z\right) = (1-z)^{c-b-a}\,{}_2F_{1}\left(\begin{matrix} c-a,\;c-b\\ c \end{matrix};\,z\right)\,,
\end{equation}
as well as of the following relation between rising and falling Pochhammer symbols:
\begin{equation}
    [x + m - 1]_n = \frac{(x)_m}{(x)_{m-n}}\,.
\end{equation} 
One arrives at the following representation:
\begin{equation}
    \nu^{(s|n)}_{\Delta} = 2^{\Delta+s-1}\,\frac{\big(\frac{d}{2}-1\big)_{s}}{\big(\Delta-1\big)_{s}}\,\frac{\big(\Delta-1\big)_{s-n}}{\big(\frac{d}{2}-1\big)_{s-n}}\,{}_2F_{1}\left(\begin{matrix} \Delta-1 + (s-n),\;\frac{d}{2}+s-1\\ \frac{d}{2} - 1 +(s - n) \end{matrix};\,-1\right)\,.
\end{equation}
If one makes use of Appell series
\begin{equation}
    F_1\left(\begin{matrix} a,\;b,\;b^{\prime}\\ c \end{matrix};\,u,v\right) = \sum_{i,j=0}^{\infty}\frac{\big(a\big)_{i+j}\big(b\big)_i\big(b^{\prime}\big)_j}{\big(c\big)_{i+j}}\,\frac{u^{j}\,v^{j}}{i!\,j!}\,,
\end{equation}
together with the following resummation formula
\begin{equation}
    \sum_{n=0}^{\infty}\frac{\big(a\big)_n\big(b^{\prime}\big)_n}{\big(c\big)_n}\,\frac{t^{n}}{n!}\,{}_2F_{1}\left(\begin{matrix} a+n,\;b\\ c+n \end{matrix};\,x\right) = F_1\left(\begin{matrix} a,\;b,\;b^{\prime}\\ c \end{matrix};\,x,t\right)\,,
\end{equation}
then the expression for $\nu^{(s|n)}_{\Delta}$ admits the following generating function
\begin{equation}
    G^{(s)}_{\Delta}(t) = \sum_{j=0}^{\infty}\nu^{(s|s-j)}_{\Delta}\,t^j = 2^{\Delta + s -1}\,\frac{\big(\frac{d}{2}-1\big)_{s}}{\big(\Delta-1\big)_{s}}\,F_1\left(\begin{matrix} \Delta-1,\;\frac{d}{2}+s-1,\;1\\ \frac{d}{2}-1 \end{matrix};\,-1,t\right)
\end{equation}
such that 
\begin{equation}
    \nu^{(s|n)}_{\Delta} =\frac{1}{(s-n)!} \left.\partial_{t}^{s-n}G^{(s)}_{\Delta}(t)\right|_{t=0}\,.
\end{equation}

\newpage

\end{document}